\documentclass{ectj}
\usepackage{natbib}
\usepackage{silence}
\usepackage{amsfonts,amssymb,graphics,epsfig,verbatim,bm,latexsym,amsmath,url,amsbsy}
\usepackage{bbm}
\usepackage{bm}
\usepackage{array}
\usepackage{blkarray}
\usepackage[style=base]{caption}
\usepackage[ruled,vlined]{algorithm2e}
\usepackage{tabularx,booktabs}
\usepackage{multirow}
\usepackage[retainorgcmds]{IEEEtrantools}
\usepackage[margin=3cm]{geometry}

\newtheorem{proposition}{Proposition}

\newtheorem{definition}{Definition}

\renewcommand{\thesection}{\arabic{section}}
\renewcommand{\theequation}{\arabic{section}.\arabic{equation}}

\newcounter{bean}

\received{...}
\accepted{...}
\volume{21}

\title[Sparse Covariance Logit Mixture]{Sparse Covariance Estimation in Logit Mixture Models}

\author[Aboutaleb et al.]{Youssef M. Aboutaleb$^{\dagger}$, Mazen Danaf$^{\dagger}$, Yifei Xie$^{\dagger}$, Moshe E. Ben-Akiva $^{\dagger}$}

\address{$^{\dagger}$ Massachusetts Institute of Technology,
            77 Massachusetts Avenue, Cambridge, MA 02139, USA.}
\email{ymedhat@mit.edu, mdanaf@mit.edu, yifeix@mit.edu, mba@mit.edu}

\def\AmSTeX{$\cal A$\kern-.1667em\lower.5ex\hbox{$\cal M$}\kern-.125em
    $\cal S$-\TeX}
\def\BibTeX{{\rm B\kern-.05em{\sc i\kern-.025em b}\kern-.08em
    T\kern-.1667em\lower.7ex\hbox{E}\kern-.125emX}}

\begin{document}
  \begin{abstract}
This paper introduces a new data-driven methodology for estimating sparse covariance matrices of the random coefficients in logit mixture models. Researchers typically specify covariance matrices in logit mixture models under one of two extreme assumptions: either an unrestricted full covariance matrix (allowing correlations between all random coefficients), or a restricted diagonal matrix (allowing no correlations at all). Our objective is to find optimal \textit{subsets} of correlated coefficients for which we estimate covariances. We propose a new estimator, called MISC, that uses a mixed-integer optimization (MIO) program to find an optimal block diagonal structure specification for the covariance matrix, corresponding to subsets of correlated coefficients, for any desired sparsity level using Markov Chain Monte Carlo (MCMC) posterior draws from the unrestricted full covariance matrix. The optimal sparsity level of the covariance matrix is determined using out-of-sample validation. We demonstrate the ability of MISC to correctly recover the true covariance structure from synthetic data. In an empirical illustration using a stated preference survey on modes of transportation, we use MISC to obtain a sparse covariance matrix indicating how preferences for attributes are related to one another. 

    \end{abstract}


 \section{Introduction}
 \setcounter{equation}{0}
The logit mixture model, also called the mixed logit model, is widely considered to be the most promising state of the art discrete choice model; see \cite{hensher2003mixed}. In a seminal paper, \cite{mcfadden2000mixed} show that under some mild regularity conditions any discrete choice model that is consistent with random utility maximization has choice probabilities that can be approximated, up to any desired precision, by a logit mixture model  through a right choice of explanatory variables and distributions for the random coefficients. The logit mixture model enables the modelling of preference heterogeneity by allowing the model's coefficients to be randomly distributed across the population under study; see Chapter 6 of \cite{train2009discrete}.

In specifying a logit mixture model, the researcher makes assumptions on the distribution of the model's coefficients (called the \textit{mixing distribution}) and on the \textit{structure} of the covariance matrix. For example, for normally distributed coefficients $\bm{\beta} \sim \mathcal{N}(\bm{\mu}, \bm{\Omega})$, the researcher decides which, if any, covariance matrix elements to estimate and which ones to constrain to zero. Typically, the researcher compares goodness-of-fit statistics on a few competing hypotheses on the structure of the covariance matrix (usually full against diagonal covariance matrix). As the number of all possible covariance matrix specifications grows super-exponentially with the number of distributed coefficients, it is not practically feasible for the researcher to comprehensively compare all possible specifications of the covariance matrix in order to determine an optimal specification to use.\footnote{If there are $R$ random coefficients in a mixed logit model, the number of different covariance matrix specifications corresponding to mutually exclusive and collectively exhaustive subsets of correlated coefficients can be determined by counting the number of ways a set of $R$ elements can be partitioned into non-empty subsets. This is given by the Bell numbers (\citealp{aigner1999characterization}). The first few Bell numbers are $1, 2, 5, 15, 52, 203, 877, 4140, 21147, 115975,\ldots$.}

In this paper, we introduce an algorithmic estimation procedure that discovers the best \textit{block diagonal} covariance matrix specification, corresponding to subsets of correlated coefficients, directly from the data. 

Parsimonious specifications of the covariance matrix are desirable since the number of covariance elements grows quadratically with the number of distributed coefficients. Consequently, sparser models provide efficiency gains in the estimation process compared to estimating a full covariance matrix. \cite{james2018estimation} demonstrates, on an empirical application, that sparser representations of the covariance matrix provide a better fit on the data than a fully unrestricted model as measured by BIC or AIC. \cite{james2018estimation} proposed a factor structured covariance approach to cast the covariances into a lower dimensional representation of latent factors. \cite{keane2013comparing} compared different logit mixture specifications with full, diagonal, and restricted covariance matrices and concluded that a full covariance matrix is not justified by the data in many cases, and that different specifications of the covariance matrix fit best on different datasets.

On the other hand, \cite{hess2017estimation} shows that ignoring statistically significant correlations between the distributed coefficients can distort the estimated distribution of ratios of coefficients, representing the values of willingness-to-pay (WTP) and marginal rates of substitution. Several studies including \cite{hess2017correlation}, \cite{kipperberg2008application}, \cite{scarpa2008utility}, \cite{revelt1998mixed}, and \cite{train1998recreation} have found statistically significant correlations between coefficients in logit mixture models on a number of empirical applications.

The conclusion is that, in general, researchers cannot know, without testing, which restrictions to impose.

In this paper, we propose a new methodology for learning the structure of the covariance matrix in logit mixture models with a normal mixing distribution (or transformations of the normal distribution such as the log-normal or Johnson SB distributions) from the data. In particular, we are interested in algorithmically identifying optimal subsets of correlated coefficients for which we estimate covariances. This corresponds to a block diagonal specification of the covariance matrix. We build on ideas from the mixed-integer optimization (MIO) framework for variable selection in linear regression models developed by \cite{bertsimas2016best}, the Hierarchical Bayes (HB) estimator for logit mixture models, (see \citealp{allenby1997introduction}, \citealp{allenby1998marketing}, and \citealp{train2009discrete}), and its extension to block diagonal covariance matrices by \cite{becker2016bayesian}. We discuss practical extensions of our proposed methodology to cases where some of the coefficients are distributed while others are fixed, latent-class logit mixture models, and logit mixtures with inter- and intra-consumer heterogeneity.

The remainder of this paper is organised as follows. Section 2 presents a brief background on covariance estimation in logit mixture models, sparse covariance matrix estimation in statistics, and the use of mixed-integer programming in model selection. Section 3 presents the proposed methodology for learning the covariance structure in logit mixture models and forms the core methodological contribution of this paper. Section 4 presents Monte Carlo simulations to validate our proposed methodology, in addition to an empirical application. Section 5 discusses the extensions of our proposed methodology mentioned above. Section 6 concludes the paper.

\section{Background}
\setcounter{equation}{0}
This section provides an overview of covariance matrix estimation in logit mixture models, the problem of estimating sparse covariance matrices in statistics, and the motivation for using a mixed-integer programming approach to select an optimal covariance matrix specification. 
\subsection{Covariance Matrix Estimation in Choice Models}
We consider the logit mixture model with the utility specification shown in equation (2.1). The indices used are $n\in \{1,2,\cdots,N\}$ for individuals, $m\in\{1,2,\cdots, M_n\}$ for choice situations (or ``menus"), and $j\in\{1,2,\cdots,J_{mn}\}$ for alternatives.
\begin{equation}
U_{jmn} = V_{jmn}(\textbf{X}_{jmn}, \bm{\zeta}_n)+\epsilon_{jmn}
\end{equation}
    	
$U_{jmn}$ is individual $n$'s unobserved utility of alternative $j$ in choice situation $m$, $V_{jmn}$ is the systematic utility function, $\textbf{X}_{jmn}$ is a vector of explanatory variables (e.g. attributes of the alternatives and characteristics of the individual), $\bm{\zeta}_n$ is a vector of individual-specific coefficients, and $\epsilon_{jmn}$ is an error term following the extreme value distribution with zero mean and unit scale. The researcher specifies a distribution for the coefficients and estimates the parameters of that distribution. Typically, a normal or log-normal distribution is specified; see \cite{ben1997modeling} and \cite{revelt1998mixed}. We assume that $\bm{\zeta}_n$ is normally distributed in the population with mean $\bm{\mu}$ and covariance matrix $\bm{\Omega}$. The $\bm{\zeta}_n$'s represent the preferences or tastes of individual decision-makers. \begin{equation}
\bm{\zeta}_n \sim \mathcal{N}(\bm{\mu}, \bm{\Omega})
\end{equation}
Exponentiation of coefficients in the utility equations is used when a particular coefficient is known to have the same sign across the population of decision-makers (e.g. a negative sign for a price coefficient) -- this is equivalent to specifying a log-normal distribution for said coefficient. 

The researcher estimates $\bm{\mu}$ and $\bm{\Omega}$. The matrix $\bm{\Omega}$ represents the covariance structure of the individual-specific coefficients. The variances on the diagonal elements reflect the magnitude of heterogeneity in these coefficients in the population, and the off-diagonal elements represent covariances between these coefficients-- indicating that preferences for one attribute are related to their preferences for another attribute; see \cite{hess2017correlation}.

Conditional on $\bm{\zeta}_n$, the probability of selecting alternative $j$ can be expressed as:

\begin{equation}
P_{jmn}(\bm{\zeta}_n)=\frac{\exp(V_{jmn}(\textbf{X}_{jmn}, \bm{\zeta}_n))}{\sum_{i=1}^{J_{mn}} \exp(V_{imn}(\textbf{X}_{imn}, \bm{\zeta}_n))}
\end{equation}

The individual-specific coefficients $\bm{\zeta}_n$ are random, and the unconditional probability of choice is obtained by integrating over the mixing distribution of these coefficients (which we assume to be normal).

\begin{equation}
P_{jmn}=\int P_{jmn}(\bm{\zeta}_n)\mathcal{N}(\bm{\zeta}_n|\bm{\mu}, \bm{\Omega})d\bm{\zeta}_n
\end{equation}

The model's parameters, $\bm{\mu}$ and $\bm{\Omega}$, can be estimated using Maximum Simulated Likelihood (MSL). This requires integration over a multidimensional distribution. Most applications using MSL for model estimation assume a diagonal covariance matrix specification, due to the computational constraints that manifest through the so-called ``curse of dimensionality": the number of draws required for simulation increases exponentially with the number of variables, making estimation highly intractable; see \cite{guevara2009estimating} and \cite{cherchi2012monte}. 

 \cite{train2005mixed} show that Bayesian methods for estimating the logit mixture model, such as Markov Chain Monte Carlo (MCMC), are less susceptible to the ``curse of dimensionality". The most commonly used estimation method of logit mixtures is an Hierarchical Bayes (HB) estimator; see \cite{allenby1997introduction}, \cite{allenby1998marketing}, and \cite{train2009discrete}. This estimator is based on a three-step Gibbs sampler with an embedded Metropolis-Hastings algorithm which we describe below:\\
 \medskip
\begin{algorithm}[ht]
\SetAlgoLined
 \setcounter{bean}{0}
       \begin{center}
\begin{list}
{\textsc{Step} \arabic{bean}.}{\usecounter{bean}} 
\item  Draw $\bm{\mu}$ given $\bm{\Omega}$ and $\bm{\zeta}_n$ using a normal Bayesian update with unknown mean and known variance. A diffuse prior is used on $\bm{\mu}$ (usually a normal distribution with zero mean and large variances).
\item Draw	$\bm{\Omega}$ given $\bm{\mu}$ and $\bm{\zeta}_n$ using a normal Bayesian update with known mean and unknown variance. Usually an Inverted Wishart (IW) prior is used on $\bm{\Omega}$ for its conjugacy property.\footnote{Common alternatives to the IW prior include the Hierarchical Inverted Wishart (HIW) which is a less informative prior than IW; see \cite{huang2013simple} and \cite{akinc2018bayesian} for details.} 
\item 	Draw $\bm{\zeta}_n$ given $\bm{\mu}$ and $\bm{\Omega}$ using the Metropolis-Hastings algorithm (the conditional posterior is proportional to logit multiplied by a normal density).	
\vspace{-\baselineskip}\mbox{}
        \end{list}
       \end{center}
 \caption{Three-step Gibbs-Sampler}
\end{algorithm}

\cite{becker2016bayesian} extended the three-step Gibbs-sampler (Algorithm 1) to allow for a block diagonal covariance structure specification. The covariance matrix is divided into mutually exclusive blocks, and each block is updated separately in step 2 of Algorithm 1. All of the off-diagonal elements that do not belong to any of the blocks are not estimated and constrained to zero.

Inferring a parsimonious covariance structure from the data has not yet been addressed in the literature on logit mixture models. On the other hand, several methods appear in the statistics literature dedicated to estimating sparse covariance matrices.

\subsection{Sparse Covariance Matrix Estimation}
In statistics, the covariance selection problem, introduced by \cite{dempster1972covariance}, involves setting some of the elements of the covariance matrix (or its inverse, the concentration matrix) to zero. Several approaches have been developed to achieve parsimonious specifications of the covariance matrix. These mainly fall into one of two categories. The first approach involves traditional multiple hypotheses testing. The relevant papers belonging to this approach include \cite{knuiman1978covariance}, \cite{porteous1985improved}, \cite{drton2004model}, \cite{drton2007multiple}, and \cite{drton2008sinful}. The second approach uses a LASSO penalty ($L_1$ norm) to achieve model selection and estimation simultaneously. The implementation of the penalty methods, however, is nontrivial because of the positive definite constraint on the covariance (or concentration) matrix; see \cite{yuan2007model}. 

\cite{yuan2007model} proposed a penalized likelihood method for model selection and parameter estimation simultaneously using an $L_1$ penalty on the off-diagonal elements. The ``maxdet" interior point algorithm from \cite{vandenberghe1998determinant} is used, and a BIC-type criterion is adopted for the selection of the tuning parameter. A similar approach was taken by \cite{banerjee2008model}, and \cite{dahl2008covariance}.
\cite{meinshausen2006high} suggested fitting a LASSO model to each variable, using others as predictors. This was later enhanced by \cite{friedman2008sparse} who proposed a fast algorithm that cycles through the variables, and fits a modified lasso regression to each variable.


In the Bayesian context, \cite{khondker2013bayesian} introduced the Bayesian Covariance Lasso (BCLASSO), which uses exponential priors on the diagonal elements and double exponential (Laplace) priors on the off-diagonal elements of the concentration matrix. The authors used Gibbs sampling to draw from the diagonal elements (since their full conditionals are available in closed form), and the standard Metropolis-Hastings algorithm for sampling from the off-diagonal elements.

\cite{wang2012bayesian} proposed a similar Bayesian estimator with similar priors (exponential priors for the diagonal elements and double exponential priors for the off-diagonal elements). Data augmentation was used to develop a more efficient block Gibbs sampler, which updates one column and row of the concentration matrix at a time.

In both of the Bayesian covariance LASSO methods by \cite{wang2012bayesian} and \cite{khondker2013bayesian} described above, the estimator does not set any of the covariance matrix entries exactly to zero, instead it generates draws that are concentrated around zero. To distinguish between zero and non-zero elements, \cite{wang2012bayesian} suggests using the thresholding approach recommended by \cite{carvalho2010horseshoe}, which compares the relative magnitudes of the penalized and non-penalized estimates. Alternatively, \cite{khondker2013bayesian} recommends using ``credible regions" based on confidence intervals of the estimates. However, the choice of the significance level is somewhat arbitrary, and coupled with the choice of the penalty itself.

There are fundamental limitations associated with the direct application of the LASSO methods just described to the logit mixture context. First, LASSO penalties penalize larger coefficients more than smaller ones. This is not desirable when estimating behavioural models, as this might result in underestimating the magnitude of heterogeneity in the population. In addition, LASSO methods not only set some covariance elements to zero, but also shrink the non-zero covariances towards zero. The estimated non-zero variances and covariances will clearly be biased, and it is not clear how a post-LASSO type methodology, (e.g. \citealp{belloni2013least}), can be applied to the logit mixture context. In contrast, our proposed mixed-integer programming methodology finds the optimal locations of zeros in the covariance matrix without penalizing the non-zero covariances during the estimation process.  

\subsection{The Mixed-Integer Programming Approach}

The problem of deciding which subsets of the distributed variables are potentially correlated (or equivalently specifying the structure of the covariance matrix to be estimated) naturally admits an integer programming formulation. A covariance element is either estimated or restricted to zero. Each of these decisions is represented by a binary decision variable in the formulation we introduce in Section 3. Making a decision on which covariances to estimate has an effect on information loss in the covariance matrix and on the likelihood, and deciding which covariances to estimate requires balancing information loss and sparsity. We solidify these ideas in the next section.

As described in Section 2.1, Algorithm 1 can be easily modified to draw from block diagonal matrices; see \cite{becker2016bayesian}. Estimating general covariance matrix structures is technically possible through the Metropolis-Hastings algorithm, but with much added computational cost. We therefore restrict the decisions on which covariances to estimate and which to constrain to zero so that the resulting covariance matrix structure is block diagonal. The covariance elements restricted to zero are not estimated, all other elements of the covariance matrix are estimated without penalty. This is in contrast to the LASSO-type methodologies which \textit{also} penalize all the covariance matrix elements leading to possible bias in the estimated values. 

Mixed-integer optimization (MIO) problems are NP-hard, which means that, in general, finding a ``certificate of optimality" in a reasonable amount of time cannot be guaranteed; see \cite{bertsimas1997introduction}.  However, massive developments in the state of the art solvers' ability to solve large scale MIO problems has enabled recent successes in applying MIO methods to statistical problems such as best subset selection \cite{bertsimas2016best}. \cite{bertsimas2016best} show that mixed-integer programming can be used to solve the best variable subset selection problem in linear regression for much larger problem sizes than what was thought possible. \cite{bertsimas2017optimal} find optimal classification trees using an MIO formulation. \cite{aboutalebmsthesis} used mixed-integer and non-convex programming techniques to find an optimal specification for nested logit models. We find that the MIO solution times in our proposed methodology take up only a small fraction of the total estimation time, with the Bayesian MCMC procedure taking up the bulk of the estimation time.

\section{Methodology}
\setcounter{equation}{0}
This section develops our proposed methodology for estimating an optimal covariance structure in logit mixture models. Our goal is to algorithmically find optimal \textit{subsets} of the distributed coefficients for which we estimate covariances.
\subsection{Optimization Problem} 
Let $\{\bm{\Omega}_d$ $d\in \mathcal{D}\}$ be a set of MCMC posterior draws from an unrestricted (full) covariance matrix $\bm{\Omega}$ (i.e. all the covariances are estimated), and let $\bm{\bm{\Psi}}$ be a sparse block-diagonal representation of the matrix $\bm{\Omega}$. In representing $\bm{\Omega}$ by a sparse matrix $\bm{\Psi}$, there is invariably some loss of information. This is represented by the following equation: 
\begin{equation}
    \bm{\Omega}=\bm{\Psi}+\mathbf{E}
\end{equation}
where the matrix $\mathbf{E}$ is the loss matrix. The problem of interest is to find the optimal balance between information, as represented by $\bm{\Psi}$, and loss as represented $\mathbf{E}$. In this section, we formulate this problem as a mixed-integer optimization problem with a quadratic objective and linear constraints. We first begin with some necessary definitions.
\begin{definition}
    A square matrix $\textbf{M}$ is \textbf{block diagonal} if its diagonal elements are square matrices of any size (possibly even $1 \times 1$), and the off-diagonal terms are zero. Formally, $\textbf{M}$ is block diagonal if there exists square matrices $\textbf{A}_1,...,\textbf{A}_m$ such that $\textbf{M}=\bigoplus_{i=1}^m\textbf{A}_i$.
Where the direct sum of any pair of matrices $\textbf{A}_{m\times n} \oplus \textbf{B}_{p \times q}$ is given as a matrix of size $(m+p)\times (n+q)$ defined as:
\begin{align*}
    \textbf{A} \oplus \textbf{B}= \begin{bmatrix} \textbf{A} & \textbf{0} \\ \textbf{0} & \textbf{B} \end{bmatrix},
\end{align*}
and the boldface zeros are blocks of zeros i.e., zero matrices.
\end{definition}
Any square matrix can be trivially considered to be block diagonal with one block. Each of the square matrices $\textbf{A}_i$ on the diagonal elements of the matrix $\textbf{M}$ represents the covariance matrix of a subset of correlated coefficients. There can be as many blocks as there are rows or columns in the original matrix $\textbf{M}$ (if all the blocks are $1 \times 1 $ matrices).

Block diagonal matrices are restrictive in the sense that the property is dependent on the particular ordering or indices of the distributed coefficients-- which is somewhat arbitrary. We would like to relax this restriction by introducing the notion of a permutation independent block diagonal matrix.
\begin{definition}
    A square matrix $\textbf{M}$ is \textbf{pseudo block diagonal} (PBD) if there exists a permutation of its indices such that the index-permuted matrix $\textbf{M}'$ is block diagonal.
\end{definition}
Note that this definition is not standard in the literature, but is necessary for our exposition. By this definition, any block diagonal matrix is trivially PBD.

We restrict our attention to PBD matrices, for two reasons. First, the three-step Gibbs-sampler (Algorithm 1) can be easily modified to draw from block diagonal matrices (\citealp{becker2016bayesian}) and by trivial extension to PBD matrices through a simple permutation of indices. The second reason is that blocks in PBD matrices naturally correspond to a partition of the correlated coefficients which are arguably more interpretable than general covariance matrices. 

At a high level, the problem of finding an optimal sparse representation of the matrix $\bm{\Omega}$ can be written as follows: 
    \begin{align}
    {\min_{\textbf{E}_d,\bm{\Psi}}} \; & \sum_{d\in\mathcal{D}}{\Vert \textbf{E}_d \Vert_2^2}\\
    {s.t. \;} & \Vert \bm{\Psi} \Vert_0 \leq k\\
    & \bm{\Omega}_d=\bm{\Psi}+\textbf{E}_d \; \forall d\in \mathcal{D} \\
    & \bm{\Psi} \textnormal{ is pseudo block diagonal (PBD)}
\end{align} 
where $\Vert \bm{\Psi} \Vert_0$ (the so-called $L_0$ norm) is the number of non-zero elements in the matrix $\bm{\Psi}$.\footnote{Technically speaking, the $L_0$ norm is not a proper norm because it is not \textit{homogeneous}. Nevertheless, this non-zero counting ``norm" appears in the statistics literature.} $k$ is a parameter that is determined through cross-validation.

Constraint (3.3) stipulates that at most $k$ elements of $\bm{\Psi}$ are nonzero and is, therefore, a sparsity constraint. (3.4) is a conservation of information constraint: what is not in the sparse representation of $\bm{\Omega}$, $\bm{\Psi}$, must be in the loss matrix $\mathbf{E}$.
Constraint (3.5) restricts the class of $\bm{\Psi}$ to PBD matrices, and can be viewed as an interpretability constraint. 
This optimization problem can be described as follows: we want to find a PBD representation of the covariance matrix with at most $k$ non-zero elements with minimal square loss of information.

To mathematically formulate constraint (3.5), we first establish an association between trees and PBD matrices. Establishing such a one-to-one correspondence enables us to use the machinery of graph theory to enforce these PBD constraints. As we will find, such a representation will enable us to write constraint (3.5) as a set of linear constraints. We discuss how to construct such a graphical representation of PBD matrices and the linear constraint representation next. 
\begin{proposition}
Any pseudo block diagonal matrix can be represented by a tree. 
\end{proposition}
\begin{proof}
Suppose $\textbf{M}$ is a PBD matrix, then by Definition 3.2 there exists some index permuted matrix $\textbf{M}'$ of $\textbf{M}$, such that $\textbf{M}'$ is block diagonal. By Definition 3.1, we can write $\textbf{M}'$ as the direct sum of square matrices, i.e., there exists matrices $\textbf{A}_1,...,\textbf{A}_m$ such that  $\textbf{M}'=\bigoplus_{i=1}^m \textbf{A}_i$.
To construct a tree representation of $\textbf{M}'$:
\begin{list}
{\textsc{Step} \arabic{bean}.}{\usecounter{bean}}   
    \item Represent each of the distributed coefficients in the model by a \textit{leaf} node ($\bullet$).
    \item Represent each of the square matrices $\textbf{A}_i$ by an \textit{internal} node ($\diamond$).
    \item Create a \textit{root }node ($\circ$) and connect it to each of the internal nodes.
    \item Connect each leaf node to the internal node representing the square matrix $\textbf{A}_i$ corresponding to its index.
\end{list}
\end{proof} \hfill$\square$\\
This procedure is illustrated by an example in Figure 1.

\begin{figure}[htbp]
\begin{minipage}[b]{0.5\linewidth}
    \centering

\includegraphics[scale=0.25]{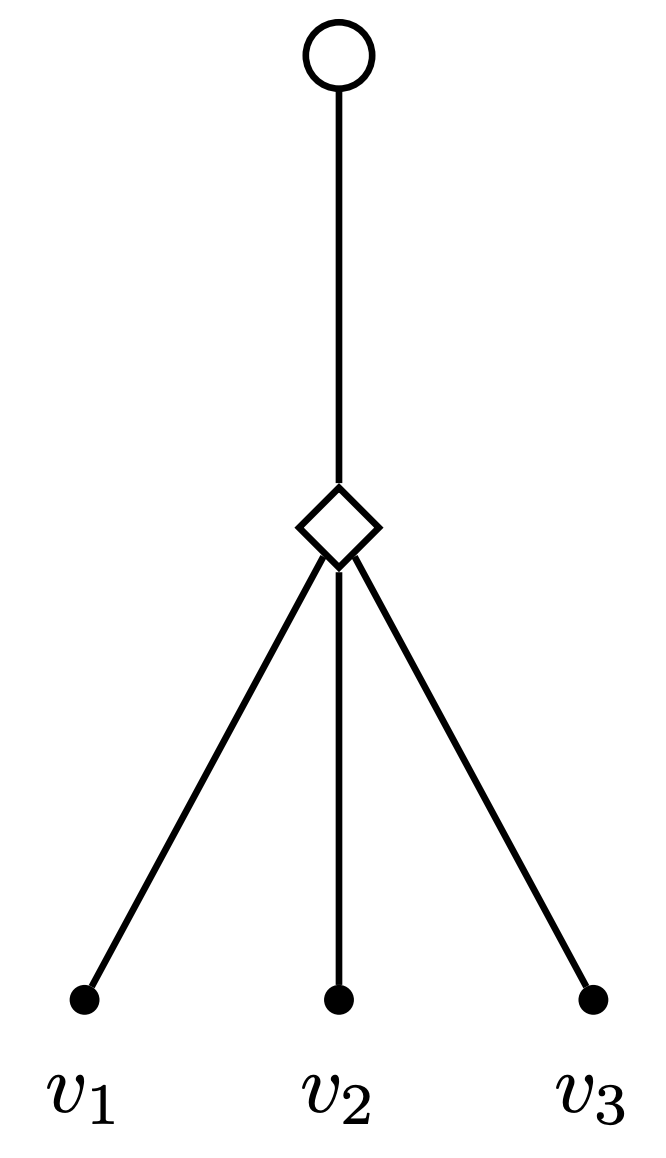}

 \[\bm{\Omega} =
\begin{blockarray}{cccc}
v_1 & v_2 & v_3  & \\
\begin{block}{(ccc)c}
  * & * & * &  v_1 \\
  * & * & * &  v_2 \\
  * & * & * &  v_3 \\
\end{block}
\end{blockarray}
 \]
       \caption*{(a) $\{\{v_1,v_2,v_3\}\}$}
    \label{fig:chapter001_dist_001} 
\end{minipage}
  \hspace{0.5cm}
\begin{minipage}[b]{0.5\linewidth}
    \centering
\includegraphics[scale=0.25]{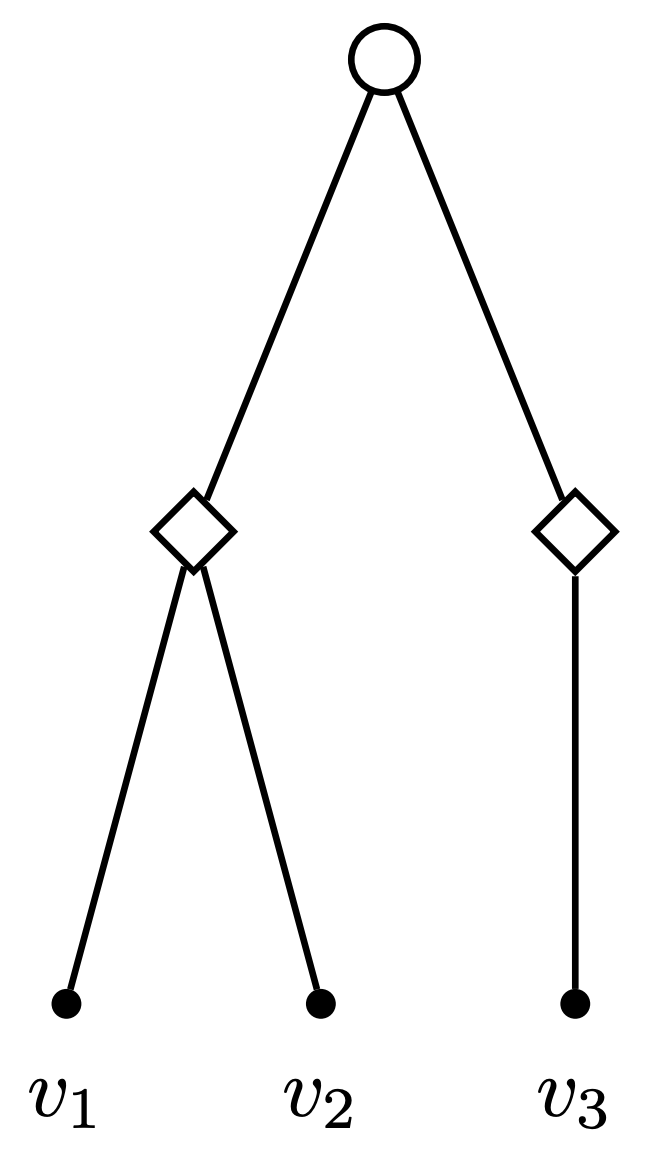}

 \[\bm{\Omega} =
\begin{blockarray}{cccc}
v_1 & v_2 & v_3  & \\
\begin{block}{(ccc)c}
  * & * & 0  &  v_1 \\
  * & * & 0 &  v_2 \\
  0 & 0 & * &  v_3 \\
\end{block}
\end{blockarray}
 \]
  \caption*{(b) $\{\{v_1,v_2\},\{v_3\}\}$}
\end{minipage}
  \hspace{0.5cm}\\
\begin{minipage}[b]{0.5\linewidth}
    \centering
\includegraphics[scale=0.25]{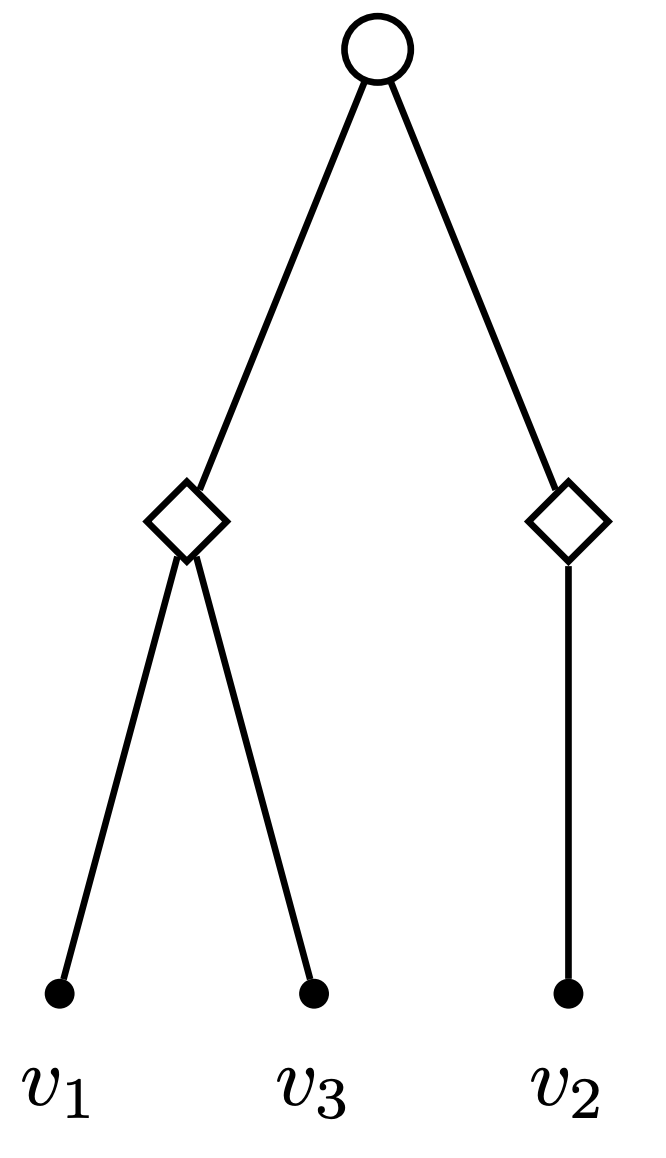}

 \[\bm{\Omega} =
\begin{blockarray}{cccc}
v_1 & v_2 & v_3  & \\
\begin{block}{(ccc)c}
  * & 0 & *  &  v_1 \\
  0 & * & 0 &  v_2 \\
  * & 0 & * &  v_3 \\
\end{block}
\end{blockarray}
 \]
   \caption*{(c) $\{\{v_1,v_3\},\{v_2\}\}$}
\end{minipage}
  \hspace{0.5cm}
\begin{minipage}[b]{0.5\linewidth}
    \centering
\includegraphics[scale=0.25]{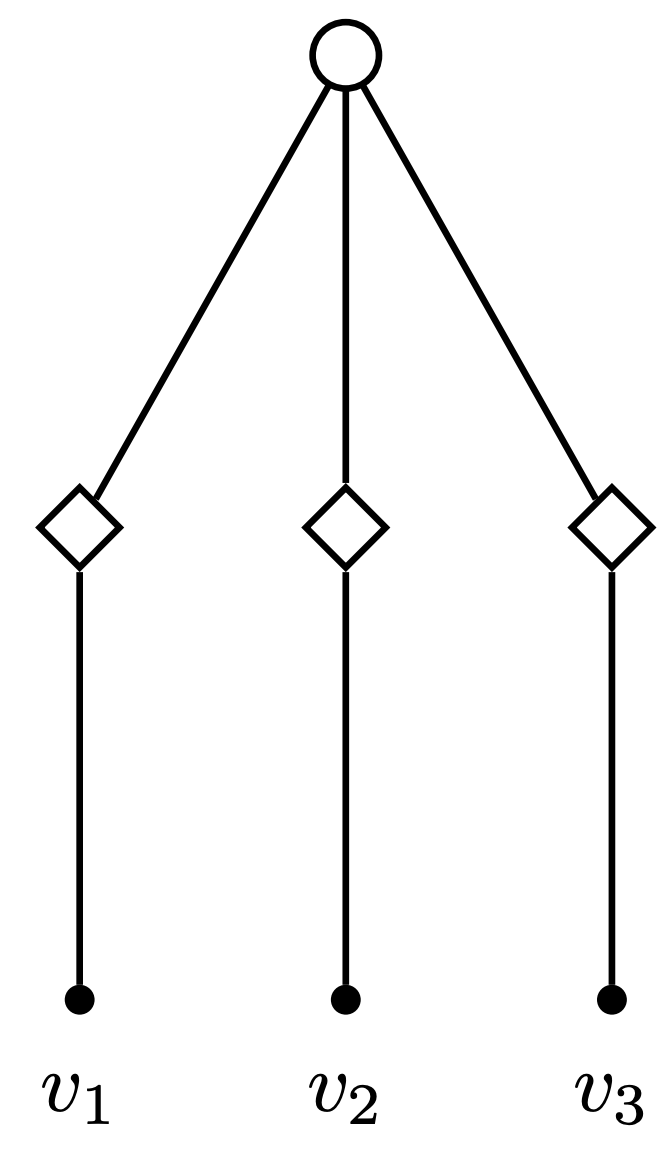}

 \[\bm{\Omega} =
\begin{blockarray}{cccc}
v_1 & v_2 & v_3  & \\
\begin{block}{(ccc)c}
  * & 0 & 0  &  v_1 \\
  0 & * & 0 &  v_2 \\
  0 & 0 & * &  v_3 \\
\end{block}
\end{blockarray}
 \]
  \caption*{(d) $\{\{v_1\},\{v_2\},\{v_3\}\}$}
\end{minipage}

      \caption{An illustration of representing PBD matrices as trees (cf. Proposition 3.1).}
    \label{fig:chapter001_reward_001}
\end{figure}
Let $\mathcal{R}$ denote the index set of distributed coefficients in the model and let $\mathcal{K}$ denote a set of $|\mathcal{R}|$ intermediate nodes representing abstract blocks. These abstract blocks correspond to the constituent block matrices $\textbf{A}_i$ of the sparse PBD representation $\bm{\Psi}$ of the covariance matrix of the distributed coefficients $\bm{\Omega}$. 

To find an optimal PBD matrix $\bm{\Psi}$, we need to determine:
\begin{enumerate}
        \item The number of blocks in $\bm{\Psi}$.
    \item How to form the blocks corresponding to subsets of distributed coefficients. 
\end{enumerate}
From the preceding discussion on representing PBD matrices as trees, we can instead consider the induced graph $\mathcal{G}=(\mathcal{V},\mathcal{{E}})$, where the set of nodes $\mathcal{V}$ consists of a root node $\{r\}$, a set of leaf nodes $\mathcal{R}$ representing the distributed variables, and a set of abstract intermediate nodes $\mathcal{K}$ representing the constituent blocks. Formally, $\mathcal{V}=\{r\}\cup \mathcal{K}\cup \mathcal{R}$.
Under this framework, the two optimization decisions posed above can be re-framed as follows:
\begin{enumerate}
    \item Which internal nodes $b\in\mathcal{K}$ to include in the graph $\mathcal{G}$? 
    \item How to form edges such that the resulting $\mathcal{G}$ is a tree?
\end{enumerate}
To this end, we define a binary variable $y_b$ for $b\in \mathcal{K}$ to denote whether or not the node representing abstract block $b$ is included in the tree representation of $\bm{\Psi}$. Furthermore, let $x_{ij}$ be a binary variable equal to one if there is a directed edge between nodes $i,j \in \mathcal{V}$, and zero otherwise. 

We now look at representing constraints (3.3) and (3.5) as linear constraints. The sparsity constraint (3.3) allows up to $k$ non-zero elements in the matrix $\bm{\Psi}$. To enforce this constraint, we first define a binary variable $z_{ij}=\mathbbm{1}\{\bm{\Psi}_{ij}\neq 0\}$ for $i,j \in \mathcal{R}$, and consider the following set of constraints:
\begin{equation}
    \underbar{\textit{M}}_{ij}z_{ij}\leq \bm{\Psi}_{ij} \leq \bar{M}_{ij}z_{ij} 
\end{equation}
Where $\underbar{\textit{M}}_{ij}$ and $\bar{M}_{ij}$ are lower and upper bounds on the size of the entries of the matrix $\bm{\Psi}$ and are determined from the MCMC posterior draws $\bm{\Omega}_d$ of the full covariance matrix $\bm{\Omega}$. The sparsity constraint (3.3) can now be represented as:
\begin{equation}
    \sum_{i,j \in \mathcal{R}} z_{ij} \leq k
\end{equation}
Since covariance matrices are symmetric, we require that:
\begin{equation}
    z_{ij}=z_{ji}
\end{equation}
To enforce constraint (3.5), we first define the following binary auxiliary variables $w_{ijb}$ to denote if variables $i,j \in \mathcal{R}$ share a common block $b\in \mathcal{K}$, that is:
\begin{equation}
    w_{ijb}=1 \iff x_{bi}=1, x_{bj}=1
\end{equation}
This can be represented by the following linear constraints:
\begin{align}
    w_{ijb} \leq x_{bi} \\
    w_{ijb} \leq x_{bj} \\
    w_{ijb} \geq x_{bi} +x_{bj}-1
\end{align}
(3.10) and (3.11) represent the forward implication of (3.9), and (3.12) represents the backward implication of (3.9).

The covariance between $i,j \in \mathcal{R}$ is estimated if and only if $i$ and $j$ share a common block. Formally,
\begin{align}
    z_{ij}=1 \iff \exists \; b \in \mathcal{K}: w_{ijb}=1
\end{align}
To represent this constraint, first note that the negation of the forward implication of (3.13) is the statement: $ \forall \; b \in \mathcal{K}, \textnormal{ } w_{ijb}=0 \implies z_{ij}=0$ (i.e., if $i$ and $j$ do not share a common block, then the corresponding covariance element $z_{ij}$ is not estimated and constrained to zero). This can be written as the following linear constraint:
\begin{align}
    z_{ij}\leq \sum_{b \in \mathcal{K}}w_{ijb}
\end{align}
Likewise, the negation of the backward implication in (3.13) is $z_{ij}=0 \implies \forall \; b\in\mathcal{K}, w_{ijb}=0$ which can be written as a set of either-or constraints:
\begin{align}
    w_{ijb} + (1-z_{ij}) \leq 1 \; \; \forall b \in \mathcal{K}
\end{align}

Finally, we need to impose a number of structural constraints as follows:
First, each of the nodes representing the distributed variables must belong to one block only:
\begin{align}
 \sum_{i \in \mathcal{K}} x_{ij}=1 \; \forall j\in \mathcal{R}   
\end{align}
Second, the edges $x_e$ must be selected such that the resulting graph is a tree. Recall that a tree with $n$ nodes must exactly have $n-1$ edges (otherwise the addition of an edge results in a cycle and the removal of an edge results in a disconnected graph). This condition can be enforced through the following equality constraint: 
\begin{align}
     \underbrace{\sum_{e \in \mathcal{{E}}}x_{e}}_\text{{Number of edges~}}=\underbrace{\big(\sum_{b \in \mathcal{K}} y_b + |\mathcal{R}| +1\big)}_\text{{Number of nodes~}} -1
\end{align}
Third, block nodes can not have connections unless included in the graph:
\begin{align}
    y_b \leq \sum_{j\in \mathcal V} x_{bj} \leq |\mathcal{V}|y_b \; \forall b \in \mathcal{K}
\end{align}
Fourth, if a block node is included, it must have an incoming connection from the root node:
\begin{equation}
    y_b=x_{rb}
\end{equation}
Finally, we always allow the diagonal elements (representing the variances) to be nonzero, i.e., $z_{ii}=1$ for all $i\in \mathcal{R}$, we permit no incoming edges to the root node:   $ x_{ir}=0 \; \; \forall i \in \mathcal{V} $, no outgoing edges from the leaf nodes $ x_{ij}=0 \; \; \forall i \in \mathcal{R}, \; j \in \mathcal{V} $, and disallow self-arcs: $ x_{ii}=0 \; \; \forall i \in \mathcal{V}$.

Putting it all together, we arrive at the following optimization problem:
\medskip
\hrule\hrule
\textbf{Pseudo Block Diagonal Optimization Problem (P)}
\hrule
    \begin{align*}
    {\min_{\textbf{E}_d,\bm{\Psi},z,x,w,y}} \; & \sum_{d\in\mathcal{D}}{\Vert \textbf{E}_d \Vert_2^2}\\
    {\textnormal{subject to:} \;}
    & \bm{\Omega}_d=\bm{\Psi}+\textbf{E}_d \; \forall d\in \mathcal{D} \\
    & \sum_{i,j \in \mathcal{R}} z_{ij} \leq k \\
     & \underbar{$M$}_{ij} z_{ij} \leq \bm{\Psi}_{ij} \leq \overline{M}_{ij} z_{ij} \; \; \forall i,j \in \mathcal{R}\\
     & z_{ij}=z_{ji} \; \; \forall i,j \in \mathcal{R}\\
     & w_{ijb} \leq x_{bi} \; \; \forall i,j \in \mathcal{R}, b \in \mathcal{K}\\
     & w_{ijb} \leq x_{bj} \; \; \forall i,j \in \mathcal{R}, b \in \mathcal{K}\\
     & x_{bi}+x_{bj} -1 \leq w_{ijb} \; \; \forall i,j \in \mathcal{R}, b \in \mathcal{K} \\
     & z_{ij} \leq \sum_{b \in \mathcal{K}}w_{ijb} \\
     & w_{ijb}+(1-z_{ij}) \leq 1 \; \forall b \in \mathcal{K}\\
     & \sum_{i \in \mathcal{K}} x_{ij}=1 \; \forall j\in \mathcal{V}\\
     & y_b \leq \sum_{j\in \mathcal V} x_{bj} \leq |\mathcal{V}|y_b \; \forall b \in \mathcal{K} \\
     & \sum_{i,j \in \mathcal{V}}x_{ij} = \big(\sum_{b \in \mathcal{K}} y_b + |\mathcal{V}| +1\big) -1 \\
     & y_b =x_{rb} \; \forall b \in \mathcal{K} \\
     & z_{ii} =1 \; \; \forall i \in \mathcal{R}\\
     & x_{ir}=0 \; \; \forall i \in \mathcal{V} \\
     & x_{ij}=0 \; \; \forall i \in \mathcal{R}, \; j \in \mathcal{V} \\
     & x_{ii}=0 \; \; \forall i \in \mathcal{V} \\
     & x \in \{0,1\}^{|\mathcal{V}| \times |\mathcal{V}|}, y\in \{0,1\}^{|\mathcal{R}|}, w \in \{0,1\}^{|\mathcal{R}| \times |\mathcal{R}| \times |\mathcal{R}|},\\
     &z \in \{0,1\}^{|\mathcal{R}| \times |\mathcal{R}|}, \textbf{E}_d \in \mathbb{R}^{|\mathcal{R}| \times |\mathcal{R}|}, \bm{\Psi} \in \mathbb{R}^{|\mathcal{R}| \times |\mathcal{R}|}
\end{align*} 
\vspace{-\baselineskip}\mbox{}
\hrule
\medskip
Given a set of draws from the full covariance matrix $\{\bm{\Omega}_d, \; d\in\mathcal{D}\}$ and a sparsity level parameter $k$, optimization problem (P) returns $\bm{\Psi}$-- a sparse PBD representation of $\bm{\Omega}$ with at most $k$ non-zeros. The sparse representation is optimal in the sense that it constructed with minimal square loss of information across draws $d\in\mathcal{D}$. 

This mixed-integer optimization problem has a quadratic objective and linear constraints and can be solved efficiently to optimality for relatively large problem sizes using standard conic quadratic programming techniques (\citealp{bertsekas1997nonlinear} and \citealp{boyd2004convex}) implemented in state of the art solvers such as GUROBI (\citealp{gurobi2015gurobi}) and CPLEX (\citealp{cplex2009v12}).
\newpage
\subsection{Overall Algorithm} 
There are two critical components to our algorithm. The first component is the optimization problem (P) described in the previous section. Given draws from the full covariance matrix, and a regularization parameter $k$, (P) outputs an optimal block structure as represented by the matrix $\bm{\Psi}$.
The second crucial component is a procedure that can estimate variance and covariance elements of pseduo block diagonal matrices from the data. This can be accomplished by applying step 2 of the three-step Gibbs sampling procedure (Algorithm 1) to each block separately as suggested by \cite{becker2016bayesian}. Tying these two components together, we arrive at the following algorithm:

\begin{algorithm}[ht]
\SetAlgoLined
 \setcounter{bean}{0}
       \begin{center}
\begin{list}
{\textsc{Step} \arabic{bean}.}{\usecounter{bean}}   

\item Estimate the logit mixture model with $\bm{\Omega}$ specified as a full covariance matrix using Algorithm 1. The output of this step includes the posterior draws of the matrix $\bm{\Omega}$: $\bm{\Omega}_d,$ $\; d\in \mathcal{D}$.
    \item For a choice of regularization parameter $k$, solve the optimization problem (P) to obtain an optimal block structure as represented by the matrix $\Psi$.
    \item Estimate the logit mixture model with $\bm{\Omega}$ specified as a block diagonal matrix according to the blocks obtained in step 2 using a block-wise three-step Gibbs sampling procedure (Algorithm 1).
    \item Evaluate the log-likelihood on training and validation data-sets.
    \item Repeat steps 2 to 4 for various values of $k$.
    \item Choose the specification with the best validation log-likelihood.
    \item Estimate the logit mixture model with optimal specification for $\bm{\Omega}$ as determined from step 6 on the full dataset.
    \vspace{-\baselineskip}\mbox{}
    \end{list}
       \end{center}
 \caption{\underline{M}ixed \underline{I}nteger \underline{S}parse \underline{C}ovariance Matrix Estimation Algorithm (\textbf{MISC})}
\end{algorithm}

The number of non-zero covariances in $\bm{\Omega}$ can vary between $|\mathcal{R}|$ to $|\mathcal{R}|^2$. Consequently, step 5 has to be repeated for $$|\mathcal{R}| \leq k \leq |\mathcal{R}|^2$$ by varying $k$ in increments of 2 (since the covariance matrix is symmetric) for a total of $|\mathcal{R}|(|\mathcal{R}| -1)/2$ runs. Note that these runs are \textit{embarrassingly parallel}-- meaning they can be run simultaneously.

We end this section by suggesting a more practical alternative to controlling sparsity than restricting the number of zero elements $\Vert \bm{\Psi} \Vert_0$.
Consider instead a modified optimization problem where Constraint (3.7) that restricts the number of non-zeros in the matrix $\bm{\Psi}$, $\sum_{i,j \in \mathcal{R}} z_{ij} \leq k$, in (P) is replaced by $\sum_{b \in \mathcal{K}} y_k \geq p$. This modified constraint now stipulates that the total number of blocks in $\bm{\Psi}$ is at least $p$. Increasing the number of blocks increases the sparsity level of the matrix since each added block constraints the covariances between the coefficients that are in the added block and those that are not to zero. 
If $\bm{\Psi}$ consists of one only block, then it is a full matrix and all covariances are estimated. With as many blocks as there are rows or columns of $\bm{\Omega}$, $\bm{\Psi}$ becomes a diagonal matrix will all the covariance elements restricted to zero. The benefit of using the number of blocks to control sparsity is that under this regime, step 5 is now to be repeated for $1 \leq p \leq |\mathcal{R}|$ by varying $p$ in increments of one for a total of $|\mathcal{R}|$ runs only. The downside being the loss in granularity in how the sparsity is specified.

\subsection{Implementation}
The MISC algorithm consisting of the mixed-integer optimization problem (P) and the block-wise three-step Gibbs sampler (Algorithm 1) has been implemented in the Julia programming language (\citealp{bezanson2017julia}) through the JuMP mathematical optimization interface (\citealp{dunning2017jump}) and the GUROBI mixed-integer solver (\citealp{gurobi2015gurobi}).  The authors make the source code accessible under an MIT licence through \url{https://github.com/ymedhat95/MISC}.

\section{Computational Experiments}
\setcounter{equation}{0}
In this section we validate the MISC algorithm introduced in Section 3. In Section 4.1, we demonstrate, through a Monte Carlo experiment that MISC can correctly recover the true covariance matrix structure. In Section 4.2, we apply our algorithm to the Apollo mode choice dataset from \cite{hess2019apollo}.
\subsection{Monte Carlo Experiments}
\subsubsection{Dataset Description}
The MISC algorithm introduced in Section 3 is applied to the synthetic choice-based-conjoint (CBC) Grapes dataset from \citet{ben2019foundations}. The setup is as follows: individuals are presented with eight menus, each including three different alternatives which are bunches of grapes with varying prices and attributes in addition to an opt-out alternative. The dependent variable is the choice between the three different bunches or not buying grapes at all (opting-out). The attributes and attribute levels are the same as in \cite{ben2019foundations}, and are presented in Table 1. The attributes of the different alternatives are drawn uniformly from their corresponding distributions.

\begin{table}
\caption{Attributes and attribute levels for the synthetic Grapes dataset.}
\centering
\begin{tabular}{ccc}
\toprule
Attribute & Symbol & Levels \\ 
\midrule
 Price & P & \$1 to \$4 \\  
 Sweetness & S & Sweet (1) or Tart (0)   \\ 
 Crispness & C & Crisp (1) or Soft (0)    \\
 Size & L & Large (1) or Small (0)    \\
 Organic & O & Organic (1) or Non-organic (0)  \\
\bottomrule
\end{tabular}
\footnotesize
\renewcommand{\baselineskip}{11pt}
\end{table}

The utility equations (normalized to the opt-out alternative) are presented in equations (4.1)-(4.2). All coefficients are distributed with inter-consumer heterogeneity, as denoted by the subscripts $n$.
\begin{align}
U_{jmn} &= e^{-\alpha_n}(-P_{jmn}+S_{jmn}\beta_{S_n}  + C_{jmn} \beta_{C_n}+L_{jmn} \beta_{L_n}  + O_{jmn}\beta_{O_n} +
SC_{jmn}\beta_{SC_n}  \nonumber \\& + SG_{jmn}\beta_{SG_n}+\beta_{q_n} )+\epsilon_{jmn}, \; \; j\in\{1,2,3\}\\
U_{opt-out,mn} &= 0 + \epsilon_{opt-out,mn} 
\end{align}
$U_{jmn}$ represents the utility of alternative $j$ in menu $m$ presented to individual $n$. $\alpha_n$  is a scale parameter. Exponentiation is used to ensure that it is always positive. $P_{jmn}$ is the price of grapes bunch $j$ in menu $m$ faced by individual $n$. Its coefficient normalized to -1. $S_{jmn}$, $C_{jmn}$, $L_{jmn}$, and $O_{jmn}$ represent sweetness, crispness, size, and organic dummies of bunch $j$ as indicated in Table 1, with coefficients $\beta_{S_{n}}$, $\beta_{C_n}$, $\beta_{L_n}$, and $\beta_{O_n}$  respectively.  $SC_{jmn}$ and  $SG_{jmn}$ represent interaction terms between sweetness and crispness, and sweetness and a gender dummy with coefficients $\beta_{SC_n}$, $\beta_{SG_n}$ respectively. $\beta_{q_n}$ is a constant for choosing one of the three bunches of grapes compared to opting-out. $\epsilon_{jmn}$ is an independently and identically distributed error term following the extreme value distribution with mean zero and unit scale.

The model is specified in the willingness-to-pay (WTP) space (i.e., the price coefficient is fixed to -1 and the scale parameter $\alpha_n$ is estimated). Therefore, all the coefficients represent the willingness-to-pay for their corresponding attributes. 

The population means and covariances of the coefficients of the synthetic population are shown Table 2. The covariance matrix is block diagonal and admits the tree representation shown in Figure 2 (cf. Proposition 3.1). Let the coefficients $\beta_S,\beta_C,\beta_L,\beta_O,\beta_{SC},\beta_{SG},\beta_q,\alpha$ in Table 2 correspond to indices $1,2,3,4,5,6,7,8$ respectively. The structure of the population covariance can be compactly represented as $\{1,2,3\}\{4,5\}\{6\}\{7\}\{8\}$. We will henceforth use this representation.

\begin{table}
\caption{Means and covariances of the distributed coefficients used to generate the synthetic data.}
\centering
\begin{tabular}{c c c c c c c c c c}
\toprule
Coefficient & Mean &\multicolumn{8}{c}{Covariance}  \\
\midrule
{} & {} &  $\beta_S$   &  $\beta_C$ & $\beta_L$ &  $\beta_O$ &  $\beta_{SC}$ &  $\beta_{SG}$ &  $\beta_q$ &$\alpha$ \\  
  $\beta_S$ & 1 & .9 & .4 & .2 & 0 & 0 & 0 & 0 & 0   \\ 
  $\beta_C$ & .3 & .4 & .2 & .05 & 0 & 0 & 0 & 0 & 0   \\  
 $\beta_L$ & .2 & .2 & .05 & .15 & 0 & 0 & 0 & 0 & 0   \\  
 $\beta_O$ & .1 & 0 & 0 & 0 & .4 & .2 & 0 & 0 & 0   \\ 
 $\beta_{SC}$ & 0 & 0& 0 & 0 & .2 & .3 & 0 & 0 & 0   \\ 
 $\beta_{SG}$ & .1 & 0 & 0 & 0 & 0 & 0 & .4 & 0 & 0   \\ 
 $\beta_q$ & 2 & 0 & 0 & 0 & 0 & 0 & 0 & 1 & 0   \\ 
 $\alpha$ & -1.5 & 0 & 0 & 0 & 0 & 0 & 0 & 0 & .25   \\ 
\bottomrule
\end{tabular}
\footnotesize
\renewcommand{\baselineskip}{11pt}
\end{table}

\begin{figure}
    \centering
\includegraphics[scale=0.35]{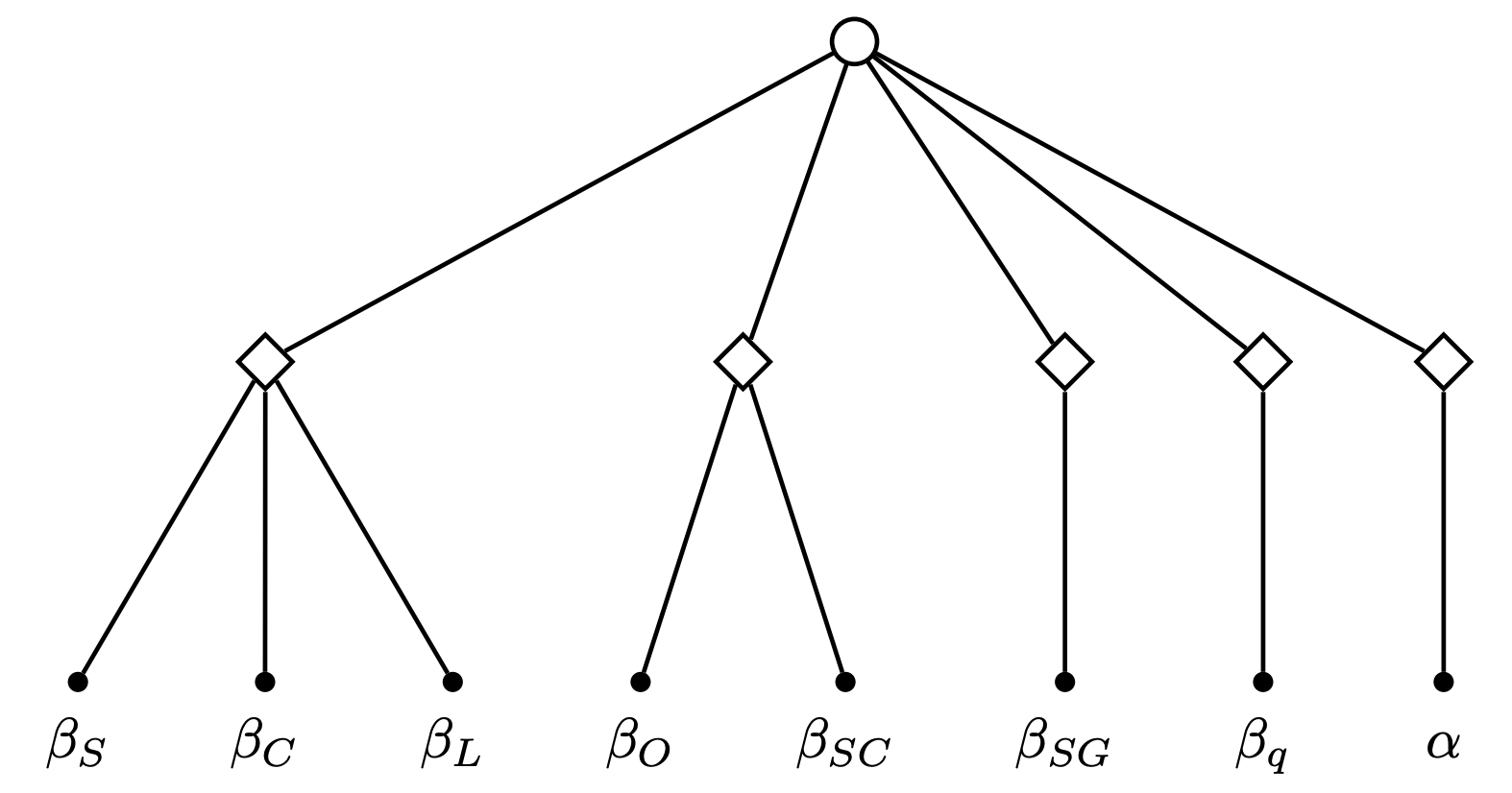}
\caption{Tree representation of the block diagonal structure in Table 2.}
\end{figure}
  
\subsubsection{Block Structure Recovery and the Effect of Sample Size}
The goal is here to use the MISC algorithm to recover the true covariance structure shown in Table 2 from the data. Out-of-sample validation is performed, to determine the optimal sparsity level, using a similar dataset, but with different individuals whose preferences follow the same multivariate normal distribution with the coefficients shown in Table 2. 

The results of the experiment are presented in Table 3 and Table 4. Table 3 shows values of the training and validation log-likelihoods for various values of the regularizer $k$, the maximum number of non-zero elements in the sparse representation of the matrix, and the corresponding optimal block structure output of the optimization problem (P). The experiment is repeated for various sample sizes. Notice that the training log-likelihood increases with \textit{decreasing} sparsity level. This is expected as estimating additional covariance matrix elements cannot worsen the log-likelihood on the training sample. On the hold-out validation sample, however, denser covariance matrices do not necessarily perform better. This is akin to the machine learning concept of over-fitting: a more complicated model does not necessarily generalize better. Table 4 shows the estimated means and covariances for the specifications corresponding to the optimal sparsity levels for each of the three experiments determined from Table 3. 

For sample sizes of $10,000$ and $1000$ individuals, the block diagonal structure with the best validation log-likelihood is the true structure  $\{1,2,3\}\{4,5\}\{6\}\{7\}\{8\}$. For the sample size of 500 individuals, a block diagonal structure where an extra covariance term is estimated is recovered: $\{1,2,3\}\{4,5\}\{6\}\{7, 8\}$.

\begin{table}[ht]
 \caption{Output of the MISC algorithm for different levels of regularization and sample sizes.}
\begin{center}

\caption*{(a) $N_{training}=10000$ and $N_{validation}=5000$}

\includegraphics[scale=0.28]{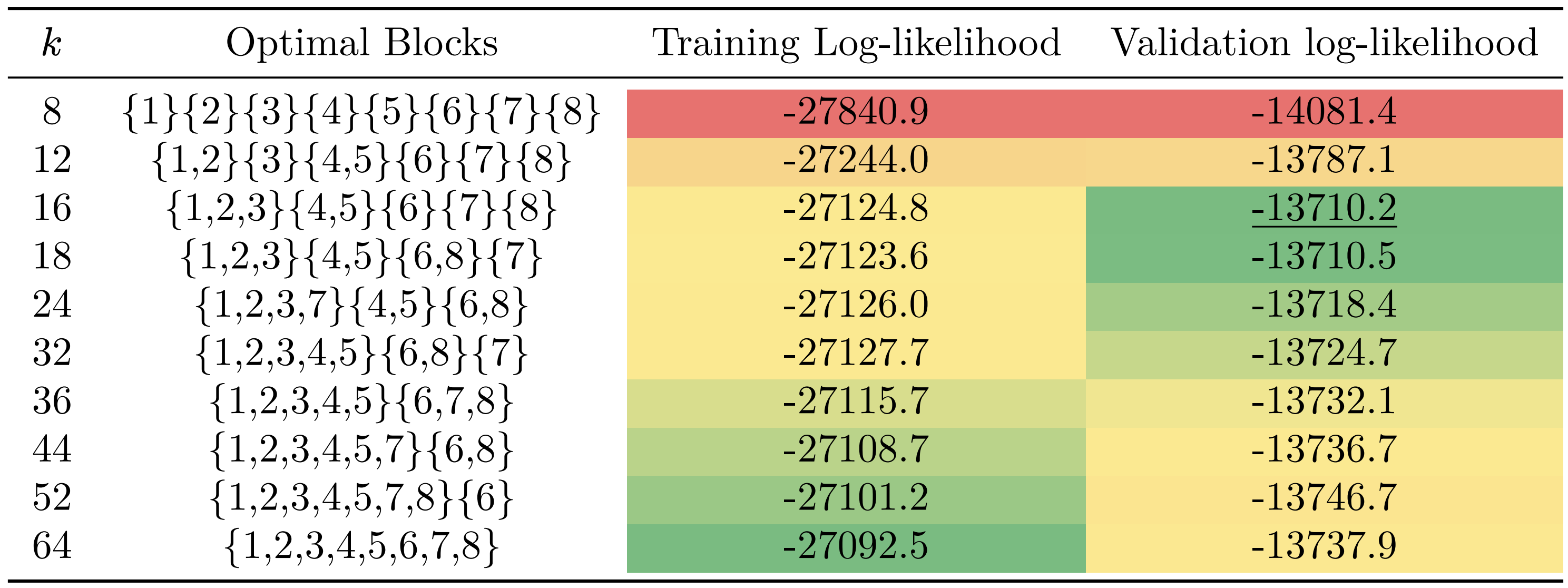}

\medskip

\caption*{(b) $N_{training}=1000$ and $N_{validation}=500$}
\includegraphics[scale=0.28]{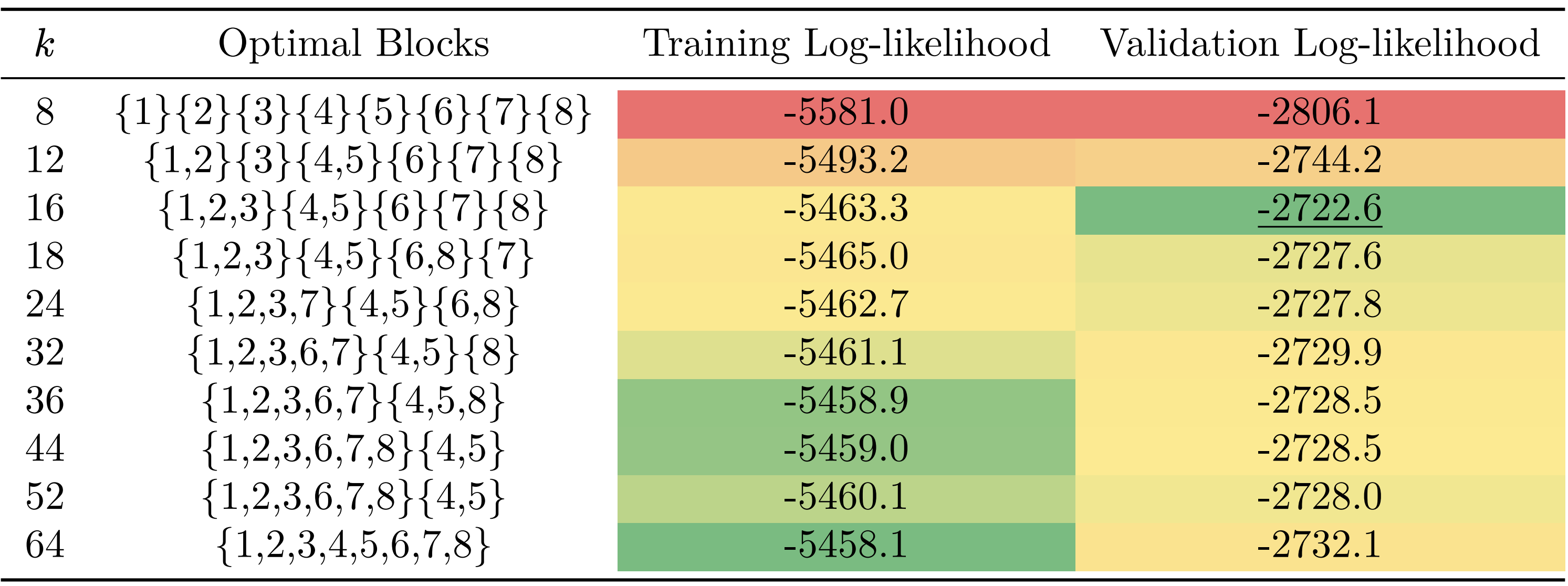}

\medskip

\caption*{(c) $N_{training}=500$ and $N_{validation}=250$}
\includegraphics[scale=0.28]{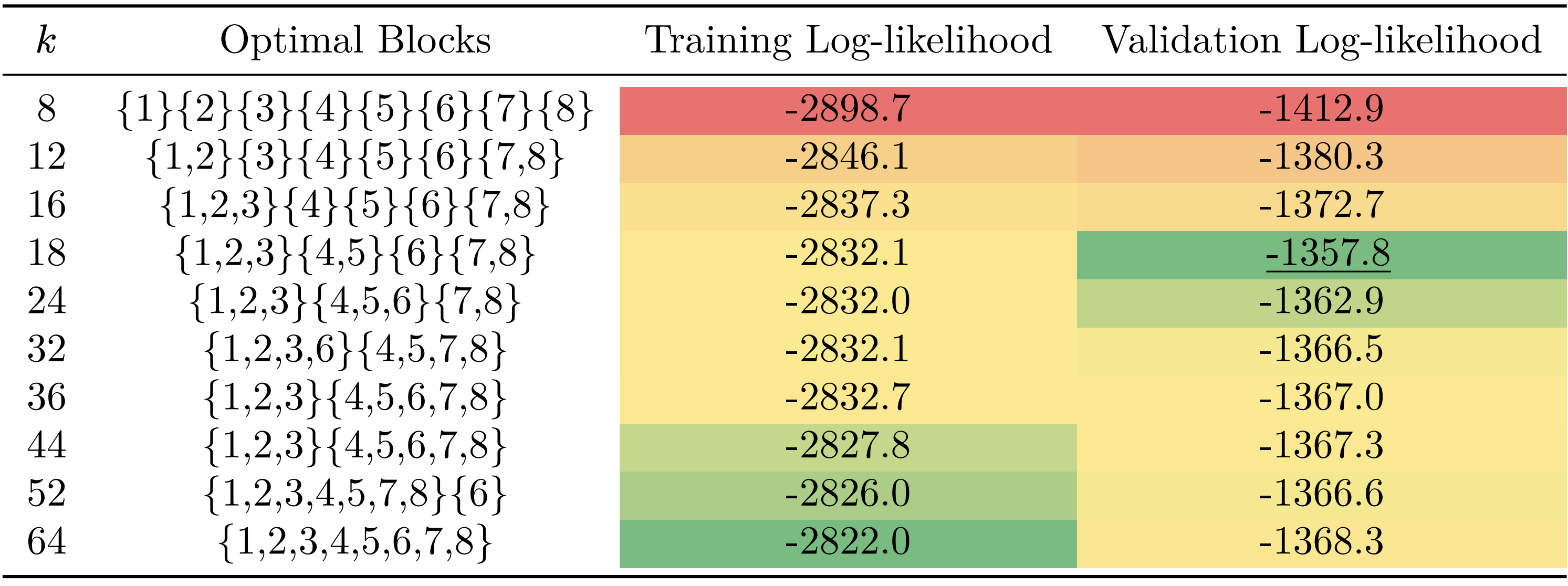}

\end{center}
\footnotesize
\renewcommand{\baselineskip}{11pt}
\textbf{Note:} {$k$ is the maximum number of non-zero elements in the covariance matrix. The Optimal Blocks column shows the optimal covariance structure at regularization level $k$-- the output of optimization problem (P). The colouring shown is proportional to the log-likelihood. The best validation log-likelihood for each experiment is underlined.}
\end{table}
\newpage
\begin{table}[ht]
\begin{center}
  \caption{Estimated parameters for the MISC optimal specifications obtained for different sample sizes.}  

\caption*{(a) $N_{train}=10000$, $N_{validation}=5000$}
\begin{tabular}{c c c c c c c c c c}
\toprule
Coefficient & Mean &\multicolumn{8}{c}{Covariance}  \\
\midrule
{} & {} &  $\beta_S$   &  $\beta_C$ & $\beta_L$ &  $\beta_O$ &  $\beta_{SC}$ &  $\beta_{SG}$ &  $\beta_q$ &$\alpha$ \\  
  $\beta_S$ & 1.000 & .905 & .394 & .195& 0 & 0 & 0& 0 & 0  \\
     $ $ & (.0147)  & (.0241)  & (.0116) & (.0086)  & -  & -  & -  & -  &   -  \\ 
  $\beta_C$ & .297 & .394 & .204& .046& 0 & 0 & 0 & 0 & 0   \\  
     $ $ & (.0097)  & (.0116)  & (.0074) & (.0050) & -  & -  & -  & -  &   -  \\ 
 $\beta_L$ & .198 & .195& .046& .154 & 0 & 0 & 0 & 0 & 0   \\  
    $ $ &  (.0070) &  (.0086) & (.0050) & (.0063)  & -  & -  & -  & -  &   -  \\ 
 $\beta_O$ & .106 & 0 & 0 & 0 & .407 & .190 & 0 & 0 & 0   \\ 
    $ $ & (.0091)  & -  & - & - &  (.0113) & (.0010)  & -  &  - &  -   \\ 
 $\beta_{SC}$ & 0.014& 0& 0 & 0 & .190& .291 & 0 & 0 & 0   \\ 
    $ $ & (.0132)  & -  & - & - & (.0010)  & (.0020)  & -  & -  &  -   \\ 
 $\beta_{SG}$ & .0910& 0 & 0 & 0 & 0 & 0 & .375 & 0 & 0   \\ 
    $ $ & (.0102)  & -  & - & - & -  &  - &  (.0131) & -  &  -   \\ 
 $\beta_q$ & 2.031 & 0 & 0 & 0 & 0 & 0 & 0 & 1.051& 0   \\ 
    $ $ & (.0159)  & - &-  & - &  - &  - &  - & (.0347)  &   -  \\ 
 $\alpha$ & -1.507 & 0 & 0 & 0 & 0 & 0 & 0 & 0 & .258  \\ 
     $ $ & (.0167)  & -  &-  & - & -  & -  & -  &  - &  (.0178)   \\ 
 \bottomrule
\end{tabular}

\medskip

\caption*{(b) $N_{train}=1000$, $N_{validation}=500$}
\begin{tabular}{c c c c c c c c c c}
\toprule
Coefficient & Mean &\multicolumn{8}{c}{Covariance}  \\
\midrule
{} & {} &  $\beta_S$   &  $\beta_C$ & $\beta_L$ &  $\beta_O$ &  $\beta_{SC}$ &  $\beta_{SG}$ &  $\beta_q$ &$\alpha$ \\  
  $\beta_S$ & .931 & .875 & .350 & .187& 0 & 0 & 0& 0 & 0  \\
     $ $ & (.0408)  & (.0757)  & (.0390) & (.0284)  & -  & -  & -  & -  &   -  \\ 
  $\beta_C$ & .241 & .350 & .178& .058& 0 & 0 & 0 & 0 & 0   \\  
     $ $ & (.0263)  & (.0390)  & (.0266) & (.0145) & -  & -  & -  & -  &   -  \\ 
 $\beta_L$ & .168 & .187& .058& .152 & 0 & 0 & 0 & 0 & 0   \\  
    $ $ &  (.0215) &  (.0284) & (.0145) & (.0226)  & -  & -  & -  & -  &   -  \\ 
 $\beta_O$ & .011 & 0 & 0 & 0 & .3607 & .151 & 0 & 0 & 0   \\ 
    $ $ & (.0273)  & -  & - & - &  (.0327) & (.0268)  & -  &  - &  -   \\ 
 $\beta_{SC}$ & .013& 0& 0 & 0 & .151& .314 & 0 & 0 & 0   \\ 
    $ $ & (.0350)  & -  & - & - & (.0268)  & (.0494)  & -  & -  &  -   \\ 
 $\beta_{SG}$ & .063& 0 & 0 & 0 & 0 & 0 & .379 & 0 & 0   \\ 
    $ $ & (.0312)  & -  & - & - & -  &  - &  (.0429) & -  &  -   \\ 
 $\beta_q$ & 2.067 & 0 & 0 & 0 & 0 & 0 & 0 & 1.259& 0   \\ 
    $ $ & (.0506)  & - &-  & - &  - &  - &  - & (.1220)  &   -  \\ 
 $\alpha$ & -1.590 & 0 & 0 & 0 & 0 & 0 & 0 & 0 & .381  \\ 
     $ $ & (.0661)  & -  &-  & - & -  & -  & -  &  - &  (.0855)   \\ 
 \bottomrule
\end{tabular}

\end{center}
\footnotesize
\renewcommand{\baselineskip}{11pt}
\end{table}

\newpage

\begin{table}[ht]
\begin{center}
  \caption*{\footnotesize{\textbf{Table 4} (continued):} Estimated parameters for the MISC optimal specifications obtained for different sample sizes.}

\caption*{(c) $N_{train}=500$, $N_{validation}=250$}
\begin{tabular}{c c c c c c c c c c}
\toprule
Coefficient & Mean &\multicolumn{8}{c}{Covariance}  \\
\midrule
{} & {} &  $\beta_S$   &  $\beta_C$ & $\beta_L$ &  $\beta_O$ &  $\beta_{SC}$ &  $\beta_{SG}$ &  $\beta_q$ &$\alpha$ \\  
$\beta_S$ & .940& .898 & .468& .170& 0 & 0 & 0& 0 & 0\\
    $ $ &(.0618) & (.1086) & (.0607) & (.0382) & - & - &- &- & - \\
$\beta_C$ &.293 & .468 & .283 & .075 & 0 & 0 & 0 & 0 & 0 \\
    $ $ &(.0475) & (.0607)& (.0518)& (.0212) & - & - & - & - & - \\
$\beta_L$ &.205 & .170 & .075 & .125 & 0 & 0 & 0 & 0 & 0 \\
    $ $ &(.0296) & (.0382) & (.0212) & (.0266) & - & -& -& - & - \\
$\beta_O$ & .153& 0 & 0 & 0 & .361 & .153 & 0 & 0 & 0\\
    $ $ &(.0385) & - &  -&- & (.0475) & (.0372) & -& -&  -\\
$\beta_{SC}$ &-.068& 0 &  0&  0 & .153 & .224 & 0 & 0 \\
    $ $ & (.0572) & -& - & - & (.0372) & (.0685) & - & - &  -\\
 $\beta_{SG}$&  .174& 0 & 0 & 0 & 0& 0 & .375 & 0& 0 \\
    $ $ &(.0451) &-  & -& - &- & -& (.0600) & - &- \\
 $\beta_q$ & 2.005 & 0 & 0& 0& 0& 0& 0 & 1.131 & .121 \\
    $ $ &(.0744) & -  & - & - & - & - & - & (.1647) & (.0788) \\
 $\alpha$ & -1.375& 0& 0 &  0& 0 & 0 & 0 & .121 & .254 \\
    $ $ &(.0783) & - &-  & -&- &- &- & (.0788) & (.0878)\\
 \bottomrule
\end{tabular}

\end{center}
\footnotesize
\renewcommand{\baselineskip}{11pt}
\textbf{Note:} {Parameter means and covariances were calculated from draws from convergent MCMC posterior draws from the mean vector and covariance matrix respectively. Standard errors are shown in brackets below the corresponding estimate.}
\end{table}

\subsubsection{Edge Cases: Full and Diagonal Matrices}
In this section, we demonstrate the ability of MISC to recover the correct covariance structure when the true covariance structure is a full matrix and when it is a diagonal matrix. For variety, we regularize using the number of blocks, $p$, instead of the number of non-zero elements (as discussed at the end of Section 3.2). 
For the full matrix experiment, a random normally distributed matrix with mean zero and standard deviation 0.1 was added to the covariance matrix in Table 2. For the diagonal experiment, only the variances in Table 2 were preserved. In both experiments the means of the coefficients were unchanged.  Data for 15000 individuals (10000 training, and 5000 validation) were generated according to these two schemes (full covariance and diagonal covariance matrices). The results are shown in Table 5. The MISC algorithm correctly recovers the structure $\{1,2,3,4,5,6,7,8\}$ for the full matrix experiment, and the structure 
$\{1\}\{3\}\{3\}\{4\}\{5\}\{6\}\{7\}\{8\}$ for the diagonal matrix experiment.

\begin{table}
\begin{center}
\caption{Output of the MISC algorithm for various level of regularization.}

\caption*{(a) Full matrix experiment}
\includegraphics[scale=0.28]{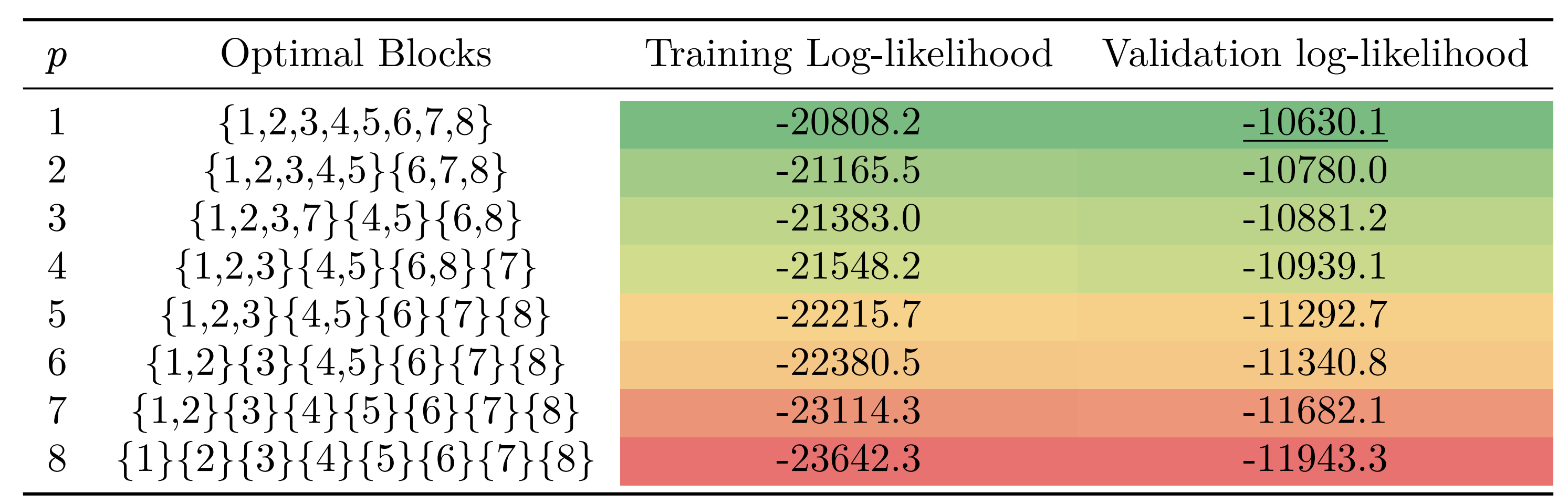}

\medskip

\caption*{(b) Diagonal matrix experiment}
\includegraphics[scale=0.28]{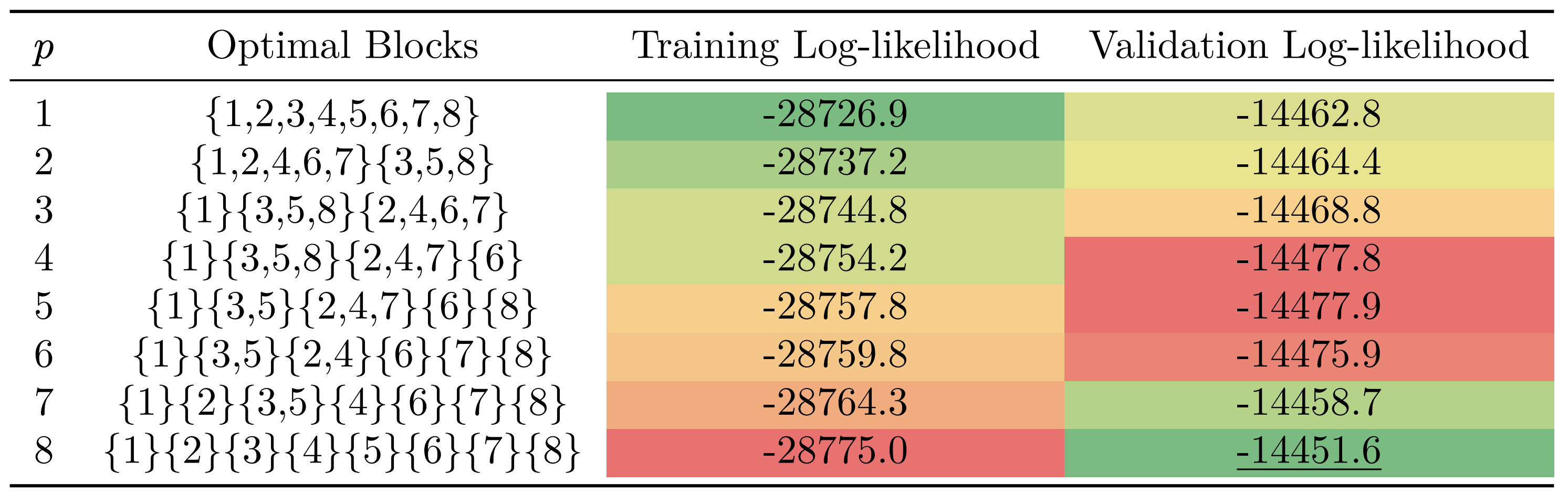}

\end{center}
\footnotesize
\renewcommand{\baselineskip}{11pt}
\textbf{Note:} {$p$ is the number of blocks in the covariance matrix. Sparsity increases with increasing number of blocks. The Optimal Blocks column shows the optimal covariance structure at regularization level $p$. The best validation log-likelihood for each experiment is underlined.}
\end{table}

\subsection{Empirical Application}

In this section we apply the MISC algorithm to the Apollo mode choice dataset from \cite{hess2019apollo}. $500$ individuals were presented with choices between four modes of transportation: car, bus, air and rail. The options were described on the basis of travel times (hours), travel costs (£), and access times (for the bus, air and rail options). Additionally for the air and bus modes, a categorical quality of service attribute was added. The quality of service attribute takes one of three levels: \textit{no frills}, \textit{wifi available}, or \textit{food available}. Each individual was presented with 14 stated preference tasks and each task had at least two of these four modes available.

A mixed logit model for the choice setup just described is specified according to utility equations (4.3)-(4.6). There is one equation for each alternative (four in total). $n$ and $m$ index individuals and stated preference tasks respectively. The dependent variable is the individual's choice of transportation mode.
\begin{align}
    U_{car,mn}&=C_{car}-e^{\beta_{tt,car}}\textnormal{Travel-Time}_{car,mn}-e^{\beta_{tc}}\textnormal{Travel-Cost}_{car,mn} +\epsilon_{car,mn}\\
    U_{bus,mn}&=C_{bus}-e^{\beta_{tt,bus}}\textnormal{Travel-Time}_{bus,mn}-e^{\beta_{tc}}\textnormal{Travel-Cost}_{bus,mn} \nonumber \\&-e^{\beta_{at}}\textnormal{Access-Time}_{bus,mn}+\epsilon_{bus,mn}\\
    U_{air,mn}&=C_{air}-e^{\beta_{tt,air}}\textnormal{Travel-Time}_{air,mn}-e^{\beta_{tc}}\textnormal{Travel-Cost}_{air,mn} +\beta_{no frills}\mathbbm{1}{\{no frills\}} \nonumber\\& +\beta_{wifi}\mathbbm{1}{\{wifi\}}+\beta_{food}\mathbbm{1}{\{food\}}-e^{\beta_{at}}\textnormal{Access-Time}_{air,mn}+\epsilon_{air,mn}\\
    U_{rail,mn}&=C_{rail}-e^{\beta_{tt,rail}}\textnormal{Travel-Time}_{rail,mn}-e^{\beta_{tc}}\textnormal{Travel-Cost}_{rail,mn} +\beta_{no frills}\mathbbm{1}{\{no frills\}}\nonumber\\& +\beta_{wifi}\mathbbm{1}{\{wifi\}}+\beta_{food}\mathbbm{1}{\{food\}}-e^{\beta_{at}}\textnormal{Access-Time}_{rail,mn}+\epsilon_{rail,mn}
\end{align}
The specification shown includes alternative-specific travel time coefficients and constants. The cost, access time, and level of service coefficients are shared across the alternatives. The coefficients of the model are assumed to be randomly distributed according to a normal distribution for which we estimate its mean and covariance matrix. The epsilon errors are $i.i.d$ extreme-valued with mean zero and unit scale. $C_{car}$ is normalized to zero for identification. The other alternative specific constants represent base preferences over the car alternative. Similarly, $\beta_{no frills}$ is normalized to zero and $\beta_{wifi}$ and $\beta_{food}$ measure the effects of additional quality of service over the \textit{no frills} option. Negation and exponentiation are used to ensure that the effects of time and cost on the utility are negative (e.g. the coefficient $\beta_{access}$ enters the utility equations as $-e^{\beta_{access}}$. This is equivalent to specifying a log-normal distribution for $\beta_{access}$). 

Table 6 shows the estimated parameters (mean and covariance matrix) corresponding to the optimal structure as determined by the MISC algorithm. Table 7 shows the output of the MISC algorithm for various levels of regularization. The specification corresponding to the best validation log-likelihood was chosen, and the corresponding estimated parameters are shown in Table 6. The optimal structure has four blocks, and allows covariances between the four alternative-specific travel time coefficients, the access time coefficient, the travel cost coefficient and the bus and rail modes alternative-specific constants. All other covariances are constrained to zero.

\begin{table}[ht]
\caption{Estimated parameters for the optimal covariance specification for the Apollo mode choice dataset.}
\begin{center}
\setlength\tabcolsep{3pt}
\begin{tabular}{ccccccccccccc}
\toprule
Coefficient & Mean& \multicolumn{11}{c}{Covariance}  \\
\midrule
 & & $C_{bus}$ & $C_{air}$& $C_{rail}$ & $\beta_{tt,car}$ & $\beta_{tt,bus}$ &$\beta_{tt,air}$& $\beta_{tt,rail}$ &$\beta_{at}$& $\beta_{wifi}$&$\beta_{food}$     & $\beta_{tc}$\\
 $C_{bus}$ & -1.23   & 5.54& 0 & -1.89 & 4.97   & 5.30  & 6.82     & 6.35  & 2.79  & 0 & 0 & 3.74  \\
         & (.419) & (1.35)& - &(.333)& (.730) & (.799) & (1.11) & (.947)& (.598) & -  & -  & (.514)\\
$C_{air}$&  .386   & 0 & .550& 0 & 0 & 0 & 0& 0 & 0& 0& 0& 0\\
        & (.153) & -& (.092)& - & - & - & - & -  & -& -& -& - \\
$C_{rail}$  &  -.428   & -1.89 & 0 & 1.04 & -2.02 & -2.11 & -2.71& -2.54& -1.11 & 0 & 0 & -1.45         \\
        & (.110) & (.333) & - & (.156)& (.243) & (.250)&(.330)& (.332)&(.202)& - & -& (.197)         \\
 $\beta_{tt,car}$ & -1.01 & 4.97&  0  & -2.02  & 5.44   & 5.46 &   7.19 & 6.70    & 2.93 & 0 & 0 & 4.05   \\
               & (.143)&(.730)& - & (.243) & (.470)& (.497) & (.694) & (.568) & (.439)  & - & -  & (.354)     \\
$\beta_{tt,bus}$ & -.365   & 5.30    & 0 & -2.11 & 5.46  & 5.98 & 7.50 & 6.98& 2.95& 0 & 0 & 4.05  \\
                & (.185) & (.799)  & -  & (.250) & (.497)& (.674)  & (.741) & (.601)& (.449)& - & -& (.414)\\
$\beta_{tt,air}$ & -1.97 & 6.82 & 0 & -2.71  & 7.19  &7.50 & 10.7  & 9.27 & 3.83 & 0 & 0  &5.36 \\
                & (.242) & (1.11) & -  &(.330)  & (.694)  & (.741) & (1.38)& (.915) & (.656) & -& - & (.468)     \\
 $\beta_{tt,rail}$ & -1.97 & 6.35  & 0& -2.54  & 6.70 & 6.98 & 9.27 & 8.87& 3.71 & 0& 0  & 4.94     \\
                & (.224) & (.947) & -  & (.332) & (.568) & (.601)& (.915) & (.823)& (.520)& -& -  & (.443)     \\
$\beta_{at}$ & -.032 & 2.79& 0  & -1.11 & 2.93& 2.95 &3.83& 3.71 & 2.16& 0    & 0 &2.13     \\
          & (.149) & (.598)  & -  & (.202) & (.439) & (.449)& (.656)& (.520)& (.487) & -& - & (.303)     \\
$\beta_{wifi}$& .908  & 0& 0 & 0 & 0  & 0 & 0 & 0 & 0 & .491& 0  & 0 \\
             & (.077) & - & -   & -  & -& -   & - & -& -  & (.086) & -   & -     \\
$\beta_{food}$&.436 & 0 & 0 & 0 & 0 & 0 & 0 & 0 & 0  & 0 & .431 & 0      \\
 & (.072) & - & -    & -   & -   & -& -   & -    & -   & - & (.086)   & -  \\
$\beta_{tc}$& -2.63&  3.74 & 0 & -1.45& 4.05  &4.05   & 5.36 & 4.94& 2.13& 0 & 0 &3.44   \\
         & (.096) & (.514)  & -& (.197) & (.354)  & (.414)  & (.468) & (.443) & (.303)  & - & - & (.339)  \\
\bottomrule
\end{tabular}

\end{center}
\footnotesize
\renewcommand{\baselineskip}{11pt}
\textbf{Note:} {The mixed logit model was reestimated on the full dataset after MISC model selection procedure. Parameter means and covariances were calculated from draws from convergent MCMC posterior draws from the mean vector and covariance matrix respectively. Standard errors are shown in brackets below the corresponding estimate.}
\noindent
\end{table}
\begin{table}[ht]
\caption{Output of the MISC algorithm for the Apollo mode choice dataset.}
\begin{center}\includegraphics[scale=0.26]{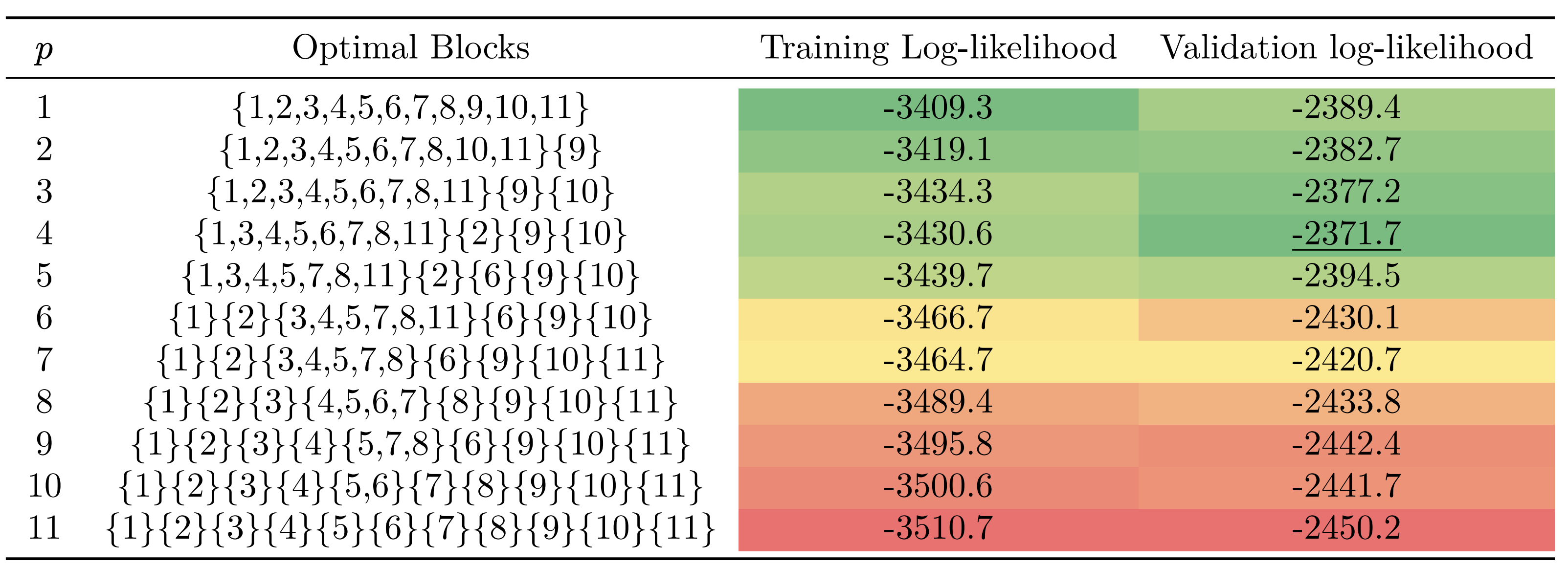}

\end{center}

\footnotesize
\renewcommand{\baselineskip}{11pt}
\textbf{Note:} {The indices shown correspond the coefficients column in Table 6. The dataset was randomly spliy into 300 individuals for training and 100 for validation. $p$ is a sparsity parameter corresponding to the minimum number of blocks in the covariance matrix. The best validation log-likelihood is underlined. The optimal block diagonal structure is obtained with $p=4$ blocks.}
\end{table}
The estimated means of the alternative-specific constants show that, \textit{ceteris paribus}, air and bus are the most and least preferred modes of transportation respectively. Travellers are also more sensitive to travel time by bus than by any other mode. The access time sensitivities outweighs, on average, the travel time sensitivities. Furthermore, the negative correlation between the alternative-specific constants of the bus and rail modes indicate that travellers who prefer to travel by rail tend to dislike travel by bus. The effects of increasing quality of service are positive, as expected, since we control for travel costs.

The estimated covariances indicate that the travel time sensitives for the different travel modes are positively correlated both among one another and with access time sensitivities, travel costs, and the preference for the bus travel mode. It appears that travellers who prefer the bus mode of travel are more sensitive to travel time and travel costs. We observe an opposite trend for travellers who prefer rail. These travellers tend to be less sensitive to travel times and travel costs. This is also reflected by the average travel time sensitivity for the rail travel mode being the lowest among the four travel modes.

Table 8 shows the calculated values of time. Travellers are more willing to spend to reduce their bus travel times than the travel times of the other available modes of transportation. Travellers are, however, most willing to spend to reduce the access times.

\begin{table}[ht]
\caption{Calculated values of time for the Apollo mode choice dataset.}
\begin{center}
\begin{tabular}{c c c}
\toprule
~                & \multicolumn{2}{l}{Value of time (£/hour)}  \\

                 & Median & Mean     
                 \\
                 \midrule
Car Travel Time  & 5.05    & 7.49                                    \\
Bus Travel Time  & 9.63    & 18.63                                    \\
Air Travel Time  & 1.93    & 10.70                                     \\
Rail Travel Time & 1.93    & 6.52                                     \\
Access Time      & 13.44   & 26.26     \\                              
\bottomrule
\end{tabular}
\end{center}

\footnotesize
\renewcommand{\baselineskip}{11pt}
\textbf{Note:} {The value of time is the extra amount of money, in £, that a traveller would be willing to pay to reduce their journey times by one hour. For an in depth demonstration of how values of time are calculated in logit mixture models see \cite{Xie2019}.}
\end{table}
\section{Extensions}
\setcounter{equation}{0}
In this section, we propose adaptations to the MISC algorithm to handle common extensions of the logit mixture model: random and fixed coefficients (Section 5.1), inter- and intra-individual heterogeneity (Section 5.2) and flexible mixing distributions (Section 5.3).
\subsection{Random and Fixed Parameters} 
A simple extension to our procedure can be made to allow optimization problem (P) to choose which variances to estimate i.e., which coefficients are random and which are fixed.
This can be done by eliminating the following set of constraints from (P):
\begin{equation}
    z_{ii}=1 \; \forall i\in \mathcal{R},
\end{equation}
and by adjusting the equality constraint in constraint (3.16) to a less than or equal to constraint:
\begin{equation}
     \sum_{i \in \mathcal{K}} x_{ij}\leq 1 \; \forall j\in\mathcal{R},
\end{equation}
Any distributed coefficients that the modified optimization problem decides not to estimate variances for, are not included in step 1 or step 2 of the 3-step Gibbs sample (Algorithm 1), but are instead estimated through a separate Metropolis-Hastings algorithm step as in \cite{khondker2013bayesian}.
This extension is particularly useful when working with small datasets where a much more parsimonious specification is required for maximal efficiency. 

\subsection{Inter- and Intra-Individual Heterogeneity} 
When multiple observations are available for each individual, it is possible to identify inter- as well as intra-individual heterogeneity, representing random taste variations among different individuals and among different choice situations of the same individual respectively. \cite{becker2018bayesian} proposed an extension to the three-step Gibbs-sampler (Algorithm 1) of logit mixture models to account for both sources of heterogeneity. The underlying model assumes that the utility equation is given by the following:
\begin{equation}
U_{jmn} = V_{jmn}(\textbf{X}_{jmn}, \bm{\eta}_{mn})+\epsilon_{jmn}
\end{equation}

$\bm{\eta}_{mn}$ represents a vector of choice-specific coefficients that follow the distribution:

\begin{equation}
\bm{\eta}_{mn} \sim \mathcal{N}(\bm{\zeta}_n, \bm{\Omega}^w),
\end{equation}

and the individual-specific means $\bm{\zeta}_{n}$ are distributed as:

\begin{equation}
\bm{\zeta}_n \sim \mathcal{N}(\bm{\mu}, \bm{\Omega}^b)
\end{equation}

$\bm{\Omega}^b$ and $\bm{\Omega}^w$ are the inter- and intra-individual covariance matrices respectively. In such applications, the proposed methodology can be extended in two different ways to account for the two types of heterogeneity:

\begin{enumerate}
    \item  Estimating sparse covariances the inter- and intra-individual covariance matrices. The MISC algorithm can be applied separately to the two covariance matrices. 

    \item  The modeler might be interested in estimating three different types of coefficients: (1) fixed coefficients, (2) coefficients with inter-individual heterogeneity only, and (3) coefficients with inter- and intra-individual heterogeneity as in \cite{becker2018bayesian}. The extension presented in Section 5.1 can be used to distinguish between the three types of coefficients. 

\end{enumerate}

\subsection{Flexible Mixing Distributions} 
The logit mixture model with normally distributed random coefficients can be extended to account for semi-parametric flexible mixing distributions. These distributions are specified as a finite mixture of normals; see  \cite{rossi2012bayesian}, \cite{bujosa2010combining}, \cite{greene2013revealing}, \cite{keane2013comparing}, and \cite{krueger2018dirichlet}. Semi-parametric distributions can overcome the major limitation of normal or lognormal mixing distributions, which is the assumption of uni-modality, as they can asymptotically mimic any shape; see \cite{vij2017random}.

The main limitation of this method is that the number of estimated coefficients is proportional to the number of ``classes" in the normal mixture. For example, a mixture of three normals entails the estimation of three covariance matrices. Our methodology can be applied to these models to enforce sparsity and reduce the number of estimated coefficients.
Let $\bm{\Omega}_{sd}$ denote draw $d\in\mathcal{D}$ for class $s\in\mathcal{S}$, $\pi_{sd}$ the fraction of the population in class $s$ in draw $d$, and $\bm{\Psi}_s$ the sparse representation of the covariance matrix for class $s$. We suggest an extension of the representation (3.2)-(3.5) to the multi-class case as follows:

 \begin{align}
    {\min_{\textbf{E}_{sd},\bm{\Psi}_s}} \; & \sum_{s \in \mathcal{S}}\sum_{d\in\mathcal{D}}\pi_{sd}{\Vert \textbf{E}_{sd} \Vert_2^2}\\
    {s.t. \;} & \sum_{s\in\mathcal{S}} \Vert \bm{\Psi}_s \Vert_0 \leq k\\
    & \bm{\Omega}_{sd}=\bm{\Psi}_s+\textbf{E}_{sd} \; \forall s\in \mathcal{S},  d\in \mathcal{D} \\
    & \bm{\Psi}_s \textnormal{ is pseudo block diagonal (PBD)} \; \forall s\in \mathcal{S}
\end{align} 
$k$ controls the sparsity level across all  $\bm{\Psi}_s$. Greater sparsity control can be achieved through an $s$-dimensional parameter $k_s \; \; s\in\mathcal{S} $ by replacing (5.7) with $\Vert \bm{\Psi}_s \Vert_0 \leq k_s \;\;\; s\in\mathcal{S}$. The clear downside being the much greater number of required runs of the MISC algorithm.
\section{Concluding Remarks}
This paper presents a new method of finding an optimal pseudo block diagonal (PBD) specification of the covariance matrix of the distributed coefficients in logit mixture models. The proposed algorithm, which we call MISC, marks a significant methodological improvement over the current \textit{modus operandi} of estimating either a full covariance matrix or a diagonal matrix. By working on PBD matrices, our method is permutation invariant in that it does not depend on the particular ordering of the distributed coefficients in the problem. The algorithm presented is an interplay between a mixed-integer optimization program and the standard MCMC three-step Gibbs-sampling procedure typically used to estimate logit mixture models. A mixed-integer program is used to find an optimal PBD covariance matrix structure for any desired sparsity level using MCMC posterior draws from the full covariance matrix. The optimal sparsity level of the covariance matrix is determined using out-of-sample validation. 

The proposed methodology is practical in that the main computational step can be completely parallelized. Furthermore, by controlling sparsity using the number of blocks in the matrix, the algorithm requires as many of the traditional MCMC runs used in logit mixture estimation as there are distributed coefficients. 

Unlike the Bayesian LASSO-based sparsity methods in the statistics literature, our method does not penalize the non-zero elements of the covariance matrix. This is desirable, in the logit mixture context, since penalizing the non-zero covariances may lead to underestimating the heterogeneity in the population under study. 

We have demonstrated the efficacy of our algorithm by applying it to a synthetic dataset where the correct block structure specification was successfully recovered. The algorithm was shown to be robust with respect to sample size. We demonstrated an empirical application to the Apollo mode choice dataset and presented a few practical extensions to our framework that are relevant for logit mixture models.  
  
    \bibliographystyle{chicago}
    
    \bibliography{references}

\end{document}